\documentclass[journal]{IEEEtran}

\usepackage{amsmath,amssymb,amsthm,citesort}
\usepackage{graphicx}
\usepackage{subfigure}
\usepackage{color}
\usepackage{algorithm}
\usepackage[noend]{algpseudocode}
\usepackage{graphics}
\usepackage{verbatim}
\usepackage{enumitem}

\newtheorem{theorem}{Theorem}

\newtheorem{remark}{Remark}

\newcommand{\argmin}{\operatornamewithlimits{argmin}}

\begin{document}

\title{Joint Energy-Bandwidth Allocation for Multi-User Channels with Cooperating Hybrid Energy Nodes}

\author{
    Vaneet~Aggarwal,  Mark R. Bell, Anis~Elgabli, Xiaodong Wang, and Shan Zhong
\thanks{V. Aggarwal, M. R. Bell, and A. Elgabli are with Purdue University, West Lafayette, IN  (e-mail:
\{vaneet, mrb, aelgabli\}@purdue.edu). X. Wang  and S. Zhong are with Columbia University, New York, NY (email: \{wangx, szhong\}@ee.columbia.edu).

This paper was presented in part at the IEEE International Conference on Communications, May 2017 \cite{pj_icc}. 
} % <-this % stops a space
}% <-this % stops a space
\maketitle

\vspace{-.2in}
\begin{abstract}
	In this paper, we consider the energy-bandwidth allocation for a network of multiple users, where the transmitters each powered by both an energy harvester and conventional grid, access the network orthogonally on the assigned frequency band. We assume that the energy harvesting state and channel gain of each transmitter can be predicted for $K$ time slots a priori. The different transmitters can cooperate by donating energy to each other. The tradeoff among the weighted sum throughput, the use of grid energy, and the amount of energy cooperation is studied through an optimization objective which is a linear combination of these quantities. This leads to an optimization problem with O($N^2K$) constraints, where $N$ is the total number of transmitter-receiver pairs, and the optimization is over seven sets of variables that denote energy and bandwidth allocation, grid energy utilization, and energy cooperation. To solve the problem efficiently, an iterative algorithm is proposed using the Proximal Jacobian ADMM. The optimization sub-problems corresponding to Proximal Jacobian ADMM steps are solved in closed form. We show that this algorithm converges to the optimal solution with an overall complexity of O($N^2K^2$). Numerical results show that the proposed algorithms can make efficient use of the harvested energy, grid energy, energy cooperation, and the available bandwidth.
	
\end{abstract}

\section{Introduction}
The rapid development of energy harvesting technologies leads to a new paradigm of wireless communications powered by renewable energy sources \cite{EHSNSI,Energy_Scav}. By their nature, some renewable energy technologies (wind, solar, and run-of-river hydro) provide intermittent generation. Thus, hybrid energy sources with a mix of energy from the grid and renewable energy sources become important. Further, the grid allows different nodes to share energy with each other. Although energy harvesting can potentially enable sustainable and environmentally friendly deployment of wireless networks, it requires efficient utilization of energy and bandwidth resources \cite{SROPPE,Zhetcomoct2015}. In this paper, we consider optimizing the weighted sum throughput of a multi-user  network while minimizing both the energy from the grid and energy cooperation between the nodes.

In the absence of conventional energy, a number of works addressed energy scheduling  with non-causal channel state information. A single-user channel is considered in \cite{FHEARS,OTPBLE,OEAWCE,TEHNFW,Zheallerton13}. For multiple users, novel scheduling algorithms have been proposed for multiple-access channels \cite{OPSMAC,ZheJSAC15}, relay channels \cite{relayenergy1,relayenergy2,6381384}, broadcast channels \cite{broadcasting,BEHRT}, and interference channels \cite{renewic,Zheisit14,Zhetcomoct2015,ZheTCOM16,ZheJSAC16}. Recently, the authors of \cite{Zheisit14,Zhetcomoct2015} considered joint allocation of energy and bandwidth for multi-user networks with renewable energy. Cooperation between nodes has been studied for a relay channel in \cite{6657835}. { Hybrid energy supply at the transmitter has been studied in \cite{6678101,6525471,6528072}, where there is a single transmission node.} Cooperation for hybrid energy supply for cellular networks has been recently considered in \cite{6874568,6871428,6888498,6517783}. However, these papers do not consider the aspect of limited capacity of battery capacity, nor the maximum transmission power constraint, which will be taken into account in this paper. Limited battery capacity helps mitigate  the energy supply variations in time  and space. For example, at any given slot, a node with sufficient  energy can either share it with other nodes with insufficient energy, or store it for future use. The maximum power constraint on the transmitting node encourages cooperation to avoid energy wastage, reduces the need of grid energy, while making a part of the incoming energy prone useless and thus has to be discharged. In this paper, we consider the cost for the use of additional grid energy, and possible energy cooperation between nodes which brings new challenges to optimize the system throughput for a multi-user system.

This paper considers seven degrees of freedom for the design which include bandwidth allocation, transmission energy allocation, local harvested energy allocation, donated energy allocation, donated energy usage allocation, grid energy allocation, and discharged energy allocation. The problem is jointly convex, but has $O(N^2K)$ variables and $O(N^2K)$ constraints for $N$ users scheduling over $K$ time slots which makes it hard for a generic convex solver. This is mainly due to the donation of energy between two users in a time-slot need to be decided, which contributes to these many variables/constraints. In this paper, we give the optimal solution with complexity $O(N^2K^2)$. In the prior work \cite{ZheJSAC15}, since the incoming energy was stored in battery, used, or discharged which helped finding the energy discharge in each time slot by a greedy algorithm. The existence of energy cooperation between nodes makes such greedy algorithm for energy discharge no longer optimal. Without energy cooperation and grid energy, the problem reduces to that in \cite{Zheisit14}. This paper gives a different optimal algorithm in this case, with the same  computational complexity of $O(NK^2)$, and also provides convergence rate. %Both the problems in \cite{Zheisit14} and in this paper are solved by iterative algorithms, however we use Proximal Jacobian ADMM \cite{2013arXiv1312.3040D} in this paper with provable $o(1/k)$ convergence rather than alternating minimization in \cite{Zheisit14}. 

%However, the proposed algorithm for the inner optimization is iterative with linear convergence while the algorithm in \cite{Zheisit14} is not iterative which is mainly because the problem structure in this paper is more complex. %The problem structure in this problem is more complex, and has multiple variables to find jointly optimal solutions without an iterative approach.

The Alternating Direction Method of Multipliers (ADMM) is a widely used algorithm for solving separable convex optimization problems with linear constraints for two sets of variables. Global convergence of ADMM was established in the early 1990’s by Eckstein and Bertsekas \cite{Eckstein:1992:DSM:153390.153393}
while studying the algorithm as a particular instance of a Douglas-Rachford splitting method. This relationship
allowed them to use the monotone operator theory to obtain their global convergence guarantees. The interest in ADMM has exploded in recent years because of applications in signal and image processing,
compressed sensing \cite{2009arXiv0912.1185Y}, matrix completion \cite{Yuan:2009vu}, distributed optimization and statistical and machine
learning \cite{Boyd:2011:DOS:2185815.2185816}, and quadratic and linear programming \cite{doi:10.1137/120878951}. Extensions to more than two blocks have been recently considered.  For example, an ADMM-type algorithm is introduced in \cite{NIPS20145256}, where during each iteration a randomly selected subset of blocks is
updated in parallel. The method incorporates a backward step on the dual update to ensure convergence.
Hong and Luo \cite{2012arXiv1208.3922H} present a convergence proof for the $n$-block ADMM when the functions are convex, but
under many assumptions that are difficult to verify in practice. The work in \cite{Han2012} shows that ADMM is convergent
in the $n$-block case when the separable functions  are strongly convex. { The separable functions in the proposed problem in this paper are not strongly convex limiting the use of this algorithm}. Recently, different extensions of ADMM like Jacobian ADMM \cite{doi:10.1137/130922793}, Flexible ADDMM \cite{2015arXiv150204391R}, Proximal Jacobian ADMM \cite{2013arXiv1312.3040D}, etc. have been considered which give  different conditions on the guarantees for convergence. 

%We find that the the sub-problem of bandwidth allocation and the rest of the sub-problems can be iteratively solved using Alternating Minimization  \cite{beck}.   Recent results have shown sub-linear convergence for this approach \cite{beck}. This iterative algorithm involves two steps - solving for bandwidth allocation and solving for the rest of the variables. The sub-problem with respect to the rest of the variables is a separable convex problem. 
In this paper, we show that recent results of Proximal Jacobian Alternating Direction Method of Multipliers (Proximal Jacobian ADMM) for any number of variables can be used to give a  convergence speed in terms of the residual error as $o(1/k)$ for $k$ iterations \cite{2013arXiv1312.3040D}. {As this algorithm is based on ADMM, it solves convex optimization problems by breaking them into smaller, easier to handle pieces which can be solved in parallel, and is thus useful in distributed scenarios \cite{Boyd:2011:DOS:2185815.2185816}. This algorithm uses  Jacobi-type scheme that helps convergence of the algorithm and adds proximal terms to get a  convergence rate of $o(1/k)$. Further, the conditions of the convergence of the algorithm are conservative, and do not require strong convexity of the separable functions.}

%The Proximal Jacobian ADMM  algorithm  solves convex optimization problems by breaking them into smaller, easier to handle pieces which can be solved in parallel, making the algorithm useful in distributed scenarios \cite{Boyd:2011:DOS:2185815.2185816}. %To the best of our knowledge, this is the first paper that applies such iterative algorithms to problems in communications with energy harvesting nodes.
 The  challenge of the non-separable objective function is handled through pairing a set of variables into a single variable. %We note that Alternating Minimization is not applied to all seven sets  of variables directly since for more than two sets of variables, strict convexity guarantees are usually needed \cite{beck} which are not satisfied by the problem under consideration. Further, ADMM is not applied to the original problem since the original problem is not separable. Using the two algorithms in the proposed manner help us to obtain the optimal solution to the original problem. To the best of our knowledge, this is the first paper that applies such iterative algorithms to problems in communications with energy harvesting nodes. The  challenge of the non-separable objective function is handled through alternating minimization, while the challenge of sum of  multi-variate non-strict convex functions is handled through ADMM. %Directly using alternating optimization over all variables may not be optimal as the results over multiple variables typically require strong convexity \cite{beck}. Similarly, ADMM requires separability. Thus, this paper considers a combination of the two iterative approaches.

 The proposed algorithm reveals a tradeoff between the system throughput, amount of energy consumption from the grid, and the amount of energy cooperation. The system designer can use this tradeoff region to choose an optimal operating point. The simulation results depict these tradeoffs. Further, we investigate the different interactions of incoming and outgoing energy, and their impact with changing cost of energy cooperation. An interesting observation is that with low cost of energy cooperation, a node with low energy arrivals may receive donated energy from other nodes not necessarily to consume it but to transfer it back when others need thus making efficient use of battery sizes at different nodes. %This is because lower incoming energy can help the node with low energy in the battery and thus energy can be donated and then taken back mitigating impacts of energy variations across nodes and time. However, the cost of energy cooperation and grid energy utilization with limited capacity battery at each node makes the tradeoff in the problem interesting.

The main contributions of the paper are as follows. 
\begin{enumerate}[leftmargin=*]
	\item This paper % is the first paper, to the best of our knowledge, that
	 { jointly} considers use of grid and renewable energy with maximum battery capacity, and energy cooperation between all nodes in a multi-user network. 
	\item %This is the first work, to the best of our knowledge,  where 
	Multi-variable Proximal Jacobian ADMM is used and shown to be optimal for this problem, with an efficient splitting of variables. 
	\item Unlike the prior works where the bandwidth allocation had to be assumed at least $\epsilon$ {(for some $\epsilon>0$)}, and an outer loop was needed to decrease $\epsilon$ {due to non-differentiability of the objective function at zero bandwidth allocation} \cite{Zhetcomoct2015,ZheJSAC16}, this paper uses power and bandwidth allocation as a single variable in Proximal Jacobian ADMM to show that such conditions are no longer needed. 
	\item The seven sets of sub-problems are solved in closed form, except one of the sub-problems in which single variable equation needs to be solved.
	\item The proposed algorithm has been used in a window-based schemes with limited prediction and the performance gap as compared to the offline strategy has been shown to decrease with increasing window size.
	\item The proposed algorithm reveals a tradeoff between the system throughput, amount of energy consumption from the grid, and the amount of energy cooperation, and performs better than the considered causal and greedy baselines. 
\end{enumerate}

  The remainder of the paper is organized as follows. In Sections II, we describe the system model
and formulate the  problem. In Section III, we solve our problem efficiently using Proximal Jacobian ADMM, and prove the convergence to the optimal solution. In Section IV, we provide numerical results for the proposed solution. Finally, Section V concludes the paper.
%In Section IV, we treat the energy-bandwidth scheduling problem for multiple non-orthogonal broadcast channels and extend the algorithms to obtain the optimal energy-bandwidth allocation.  Simulation results are provided in Section V. Finally, Section VI concludes the paper.

\vspace{-.2in}
\section{System Model and Problem Formulation}

Consider a network consisting of $N$ pairs of transmitters and receivers with a total bandwidth
of $B$ Hz. Assume that no two transmitters can transmit in the same time slot and the same
frequency band thus the channel is accessed orthogonally by sharing the total bandwidth without
any overlap. We consider a flat-fading channel where the channel gain is constant within the
entire frequency band of $B$ Hz and over the coherence time of $T_c$ seconds. Assume a scheduling
period of $K$ time slots and the duration of a time slot of $T_c$ seconds. We denote $X_{nki}$ as the symbol sent to the receiver of link $n$ at instant $i$ in slot $k$ ({ a time-slot is composed of multiple time-instants }). The corresponding  received signal at receiver $n \in {\cal N}\triangleq \{1, \cdots, N\}$  is given by
\begin{equation}
Y_{nki} = h_{nk}X_{nki} + Z_{nki},
\end{equation}
where $h_{nk}$ represents the complex channel gain for link $n$ in slot $k$,  and $Z_{nki}\sim {\sf CN}(0,1/T_c)$ is the i.i.d. complex Gaussian noise (i.e., the power spectrum density of the noise is $1/T_c$). We denote $H_n^k\triangleq |h_{nk}|^2$ and denote $p_n^k$ as the  total transmission energy consumption for link $n$ in slot $k$. Without loss of generality, we assume $T_c=B=1$. Assuming that link $n$ uses a normalized bandwidth of $a_n^k$ in time-slot $k$, we use an upper bound on the achievable rate for link $n$,   $\sum_{k=1}^{K} a_n^k \log(1+\frac{p_n^kH_n^k}{a_n^k})$ as the system performance metric \cite{Cover:2006:EIT:1146355}, where $0\cdot\log(1+\frac{x}{0}) \triangleq 0$. {  We note that in general frequency reuse may give benefits, while optimal capacity for general  interference channels is an open problem \cite{4567443}. Thus, we use bandwidth splitting to orthogonalize transmissions, while still being able to have multiple users transmit in the same time slot. Joint energy-bandwidth allocation has been shown to be advantageous with energy harvesting with limited capacity battery since multiple users can transmit simultaneously and avoid energy wastage  \cite{Zhetcomoct2015,ZheJSAC16}.}

Assume that each transmitter is equipped with a hybrid energy source, with access to both the energy from the grid, and the energy from the energy harvesting device. The energy from the grid to transmitter $n$ is unbounded, whose cost  influences the amount of energy that can be used from the grid. For transmitter $n$ in time-slot $k$, let $g_n^k$ be the energy used from the grid, and $l_n^k$ is the amount of local harvested energy that is used. The energy harvesting device at transmitter $n$ harvests energy from the surrounding environment, and is equipped with a buffer battery of capacity $B_n^{\max}$.  We denote $E_n^k$ as the total energy harvested up to the end of slot $k$ by transmitter $n$. Since in practice energy harvesting can be accurately predicted for a short period~\cite{WPDFUE,PMEHW}, we assume that the amount of the harvested energy, $E_n^k-E_n^{k-1}$ in each slot $k$ is known.  { Moreover, the short-term prediction of the channel gain in slow fading channels is also possible \cite{FCPMRA}.} Therefore, we assume that $\{H_n^k\}$ and $\{E_n^k\}$ are known non-causally before scheduling. { We will further consider a window-based scheme with limited look-ahead information  of the channels and energy harvesting  in Section IV.}

We further assume that different transmitters can donate energy to each other. { The energy donation can for instance happen through a power grid or as a wireless power transfer  \cite{6678101,6525471,6528072}.} However, there is a  cost to this energy cooperation which will prioritize using the locally harvested energy at each node as opposed to cooperation. Let $r_{n,m}^k$ be the amount of energy that is donated from node $n$ to node $m$ in time slot $k$, where $r_{n,n}^k =0 $. A part of the incoming donated energy $s_n^k \le \sum_{m=1}^N r_{m,n}^k$ is used while the rest is stored in the battery.  Thus, the amount of power used for communication for transmitter $n$ in time slot $k$ is  $p_n^k=l_n^k+ s_{n}^k+g_n^k$.

We assume that each transmitter $n$ has a maximum per-slot transmission energy consumption, $P_n$, such that $p_n^k\leq P_n$ for all $k\in {\cal K}\triangleq \{1,2,\cdots,K\}$. Thus, all the energy may not be used and some may get wasted. Let $D_n^k$ be the amount of energy that is discharged {(or wasted)} by node $n$ in time slot $k$. For transmitter $n$, assuming that the battery is empty initially, then the battery level at the end of slot $k$ can be written as
\begin{eqnarray}\label{eq:battery}
B_n^k &=& B_n^{k-1} + \left( E_n^k - E_n^{k-1}\right)-  l_n^k -\sum_{m\in {\cal N}}r_{n,m}^k \nonumber\\&&+ \sum_{m\in {\cal N}}r_{m,n}^k - s_n^k-D_n^k,
\end{eqnarray}
where $B_n^k$ must satisfy  $0\leq B_n^k \leq B_n^{\max}$ for all $k\in{\cal K}$. The  constraints on the battery level can be re-written as

\begin{eqnarray}\label{eq:rbattery}
0&\leq&  E_n^k-\sum_{t=1}^{k} l_n^t -\sum_{t=1}^{k}\sum_{m\in {\cal N}} r_{n,m}^t+\sum_{t=1}^{k}\sum_{m\in {\cal N}} r_{m,n}^t \nonumber\\
&&-\sum_{t=1}^{k}s_n^t-\sum_{t=1}^{k}D_n^t  \leq B_n^{\max}.
\end{eqnarray}
Moreover, we denote ${\cal P}\triangleq \{\boldsymbol{p}_n\;|\; \boldsymbol{p}_n\triangleq [p_n^1,p_n^2,\ldots,p_n^K], n\in{\cal N}\}$ as the {\em transmission energy allocation}, ${\cal A}\triangleq \{\boldsymbol{a}_n\;|\; \boldsymbol{a}_n\triangleq [a_n^1,a_n^2,\ldots,a_n^K], n\in{\cal N}\}$ as the {\em bandwidth allocation}, ${\cal L}\triangleq \{\boldsymbol{\cal L}_n\;|\; \boldsymbol{l}_n\triangleq [l_n^1,l_n^2,\ldots,l_n^K], n\in{\cal N}\}$ as the {\em local harvested energy allocation}, ${\cal R}\triangleq \{\boldsymbol{r}_{n,m}\;|\; \boldsymbol{r}_{n,m}\triangleq [r_{n,m}^1,r_{n,m}^2,\ldots,r_{n,m}^K], n\ne m, n,m\in{\cal N}\}$ as the {\em donated energy allocation}, ${ \cal S}\triangleq \{\boldsymbol{s}_n\;|\; \boldsymbol{s}_n\triangleq [s_n^1,s_n^2,\ldots,s_n^K], n\in{\cal N}\}$ as the {\em donated energy usage allocation}, ${ \cal G}\triangleq \{\boldsymbol{g}_n\;|\; \boldsymbol{g}_n\triangleq [g_n^1,g_n^2,\ldots,g_n^K], n\in{\cal N}\}$ as the {\em grid energy allocation}, and ${\cal D}\triangleq \{\boldsymbol{D}_n\;|\; \boldsymbol{D}_n\triangleq [D_n^1,D_n^2,\ldots,D_n^K],n\in{\cal N}\}$ as the {\em discharged energy allocation}. {A system model is described in Figure \ref{fig:model}, where we consider four transmitters, each equipped with a battery. For transmission, a part of grid energy, energy from battery, renewable energy, and donated energy is used. All the arriving energy that could not be used in a time slot is saved in the battery. The centralized controller makes all decisions of the different allocations for each node in each time slot. }

\begin{figure}[ht]
	\vspace{-.2in}
	\begin{center}
		\includegraphics[trim = .1in 1.8in 0.1in 1.8in, clip, width= 0.45\textwidth]{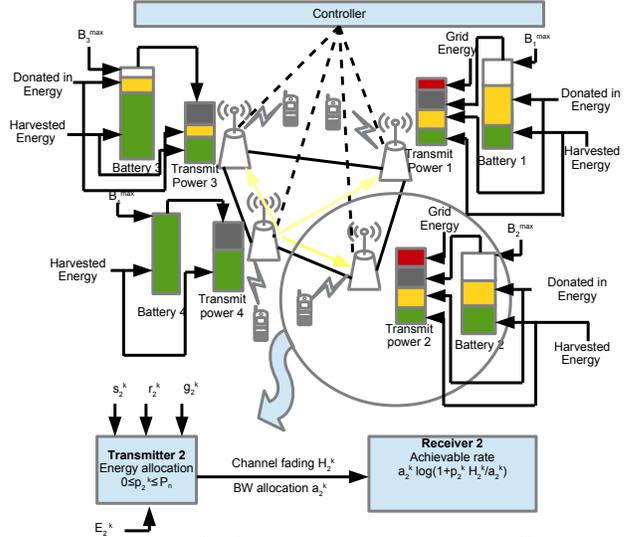}
		\vspace{-5mm}
		\caption{\small System model for four transmitters depicting different energy arrivals at a node.}
		\label{fig:model}
	\end{center}
	\vspace{-.2in}
\end{figure}

We wish to maximize the weighted sum rate of all links while minimizing the use of grid energy and energy cooperation. Thus, the objective is to maximize

\vspace{-.1in}
\begin{eqnarray}
&&C_{\cal W}({\cal P},{\cal A},{\cal L},{\cal R}, {\cal S}, {\cal G}, {\cal D}) \nonumber\\&=&\sum_{n=1}^{ N} W_n \sum_{k=1}^{K} a_n^k \log(1+\frac{p_n^kH_n^k}{a_n^k}) - \lambda\left(\sum_{k=1}^{K}\sum_{n=1}^{N} g_n^k\right)\nonumber\\&& -\mu\left(\sum_{k=1}^{K}\sum_{n=1}^N \sum_{m=1, m\ne n}^N r_{n,m}^k\right),
\end{eqnarray}
where  $\lambda\ge 0$ and $\mu\ge 0$ are parameters that impact the cost of using energy from the grid, and cooperation, respectively, and ${\cal W}\triangleq \{W_n, n\in{\cal N}\}$ is the {\em weight set}, which determines the weight (or priority)  of different links. Increasing $\lambda$ would mean that the energy from the grid will be used less since it gets more expensive. Similarly, increasing $\mu$ would make energy cooperation more expensive. {The values of $\lambda$ and $\mu$ can be chosen by the system designer to choose a tradeoff point between the system throughput, amount of energy cooperation, and the use of grid energy. %A tradeoff with change of $\mu$ is depicted in Figure \ref{fig:tradeoff}.}

%\begin{figure}[ht]
%	\vspace{-2mm}
%	\begin{center}
%		\includegraphics[trim = 2in .85in 4.3in 1.9in, clip, width= 0.3\textwidth]{tradeoff.pdf}
%		\vspace{-.15in}
%		\caption{\small Tradeoff between weighted sum rate and amount of cooperation.}
%		\label{fig:tradeoff}
%	\end{center}
%	\vspace{-.4in}
%\end{figure}

%the system is willing to pay higher cost and get more grid (and/or donated) energy in order to achieve higher rate, and .  In particular, when $W_m=1$ for all $m\in{\cal M}$ and $\lambda$ and $\mu=0$, $C_{\cal W}({\cal P},{\cal A}) $ becomes the throughput of the network.

%Note that, all $a_m^k$, $p_m^k$, $l_m^k$, $r_{n,m}^k$, and $g_m^k$ can be zero in \eqref{eq:oobj}. However, if $a_m^k=0$, the channel rate of link $m$ in slot $k$ is zero, even if the energy allocation $p_m^k>0$, thus $p_m^k$ is actually wasted. However, we still treat the pair $(a_m^k=0,p_m^k >0)$ as feasible as long as $\sum_{m\in{\cal M}_n}p_m^k \leq P_n$.

We thus have the following problem to optimize the system resources:
\begin{equation}\label{eq:oproblem}
\max_{{\cal P},{\cal A},{\cal L},{\cal R}, {\cal S},{\cal G}, {\cal D}}C_{\cal W}({\cal P},{\cal A},{\cal L},{\cal R}, {\cal S},{\cal G}, {\cal D}),
\end{equation}
subject to

\begin{equation}\label{eq:ocst}
\left\{\begin{array}{l}
0\leq  E_n^k-\sum_{t=1}^{k} l_n^t -\sum_{t=1}^{k}\sum_{m\in {\cal N}} r_{n,m}^t\\
+\sum_{t=1}^{k}\sum_{m\in {\cal N}} r_{m,n}^t -\sum_{t=1}^{k}s_n^t-\sum_{t=1}^{k}D_n^t  \leq B_n^{\max},\\
p_n^k -l_n^k- s_{n}^k - g_n^k = 0, \  s_n^k \le \sum_{m\in {\cal N}, m\ne n} r_{m,n}^k,\\
\sum_{i=1}^N a_i^k = 1, p_n^k \leq P_n, \\
a^k_m, p_n^k, l_n^k, r_{n,m}^k \text { for } n\ne m, s_n^k, g_n^k,D_n^k \geq 0,
\end{array}\right.
\end{equation}
for all $k\in{\cal K}, m, n\in{\cal N}$.

{
\begin{remark} We note that the energy donated from one node to another may lead to an energy transfer loss {(e.g., propagation loss)}, which can be easily incorporated into the model by adding an efficiency parameter to the received energy. This will not change any of the results in this paper, and is thus ignored for the rest of this paper.
\end{remark}
}

{ We note that the problem is convex with respect to each set of variables. However, it is not separable in all the variables. Further, the objective is not differentiable at $a_n^k=0$ which limits the applicability of a generic convex solver. In addition, the complexity of a generic convex solver is exponential in the number of constraints \cite{Boyd:2004:CO:993483}, which in this case is $O(N^2K)$. Thus, this paper proposes an algorithm that exploits the problem structure to give a computationally efficient solution.}
\vspace{-.2in}
\section{Optimal Algorithm}
\vspace{-.05in}
There are seven sets of variables in the problem - bandwidth allocation ${\cal A}$, transmission energy allocation ${\cal P}$, local harvested energy allocation ${\cal L}$, donated energy allocation ${\cal R}$, donated energy usage allocation  ${ \cal S}$, grid energy allocation ${ \cal G}$, and discharged energy allocation ${\cal D}$. We note that the proposed problem in \eqref{eq:oproblem}-\eqref{eq:ocst} is jointly convex in all variables. We  use the Proximal Jacobian Alternating Direction Method of Multipliers (Proximal Jacobian ADMM) technique to solve this problem \cite{2013arXiv1312.3040D}. The Proximal Jacobian ADMM  algorithm  solves convex optimization problems by breaking them into smaller and easier pieces which can be run in parallel and is thus useful for large-scale distributed convex optimization.  Since the standard Proximal Jacobian ADMM does not allow inequalities, we add additional variables to only have equality constraints.

\begin{flalign*}
&\Gamma({\cal P},{\cal A},{\cal L},{\cal R},{\cal S},{\cal G}, {\cal D}, {\cal U}) =  -C_{\cal W}({\cal P},{\cal A},{\cal L},{\cal R}, {\cal S},{\cal G}, {\cal D}) +  \nonumber \\
&\sum_{n=1}^N \sum_{k=1}^K\left(I(u_{1,n}^k)+I(u_{2,n}^k)+I(u_{3,n}^k)+I(u_{4,n}^k)+I(p_n^k)+I(\frac{}{}\right. \nonumber\\
&a_n^k)+I(l_n^k)+I(s_n^k)+I(g_n^k)+I(D_n^k) \left.+\sum_{m=1, m\ne n}^N I(r_{n,m}^k)\right), 
\end{flalign*}
where ${\cal U}\triangleq \{(\boldsymbol{u}_{1,n},\boldsymbol{u}_{2,n},\boldsymbol{u}_{3,n},\boldsymbol{u}_{4,n})| \boldsymbol{u}_{i,n}\triangleq [u_{i,n}^1,u_{i,n}^2,\ldots,u_{i,n}^K], n\in{\cal N}, i\in \{1,2,3,4\}\}$ are the auxiliary variables that help remove inequalities in the constraints, and $I(\cdot)$ is the indicator function which represents $I(x)=0$ for $x\geq0$ and is infinite otherwise. Thus, the problem in \eqref{eq:oproblem}-\eqref{eq:ocst} becomes

\begin{equation}
\quad \min_{{\cal P},{\cal A},{\cal L},{\cal R},{\cal S},{\cal G}, {\cal D}, {\cal U}} \Gamma({\cal P},{\cal A},{\cal L},{\cal R},{\cal S},{\cal G}, {\cal D}, {\cal U}),\label{epdef}
\end{equation}
subject to
\begin{equation}\label{epconst}
\left\{\begin{array}{l}
\sum_{t=1}^{k} l_n^t +\sum_{t=1}^{k}\sum_{m\in {\cal N}, m\ne n} r_{n,m}^t +\sum_{t=1}^{k}D_n^t+u_{1,n}^k\\-\sum_{t=1}^{k}\sum_{m\in {\cal N}, m\ne n} r_{m,n}^t +\sum_{t=1}^{k}s_n^t=E_n^k,\\
\sum_{t=1}^{k} l_n^t +\sum_{t=1}^{k}\sum_{m\in {\cal N}, m\ne n} r_{n,m}^t +\sum_{t=1}^{k}D_n^t-u_{2,n}^k\\-\sum_{t=1}^{k}\sum_{m\in {\cal N}, m\ne n} r_{m,n}^t +\sum_{t=1}^{k}s_n^t= E_n^k-B_n^{\max},\\
p_n^k -l_n^k- s_{n}^k - g_n^k = 0, p_n^k +u_{3,n}^k = P_n,  \\s_n^k +u_{4,n}= \sum_{m\in {\cal N}, m\ne n} r_{m,n}^k,
\sum_{i=1}^N a_i^k = 1,
\end{array}\right.
\end{equation}
for all $k\in{\cal K}, m, n\in{\cal N}$.

Let the augmented Lagrangian $\psi$ be defined as in Equation \eqref{eq:augmented}. The Proximal Jacobian ADMM updates variables $(p_n^k,a_n^k)$, $l_n^k$, $r_{m,n}^k$, $s_n^k$, $g_n^k$, $D_n^k$, and $u_n^k$ for all $m, n, $ and $k$ in sequence, whose steps are summarized in Algorithm 1. Note that the objective function is not separable in $p_n^k$ and $a_n^k)$, which is why $(p_n^k,a_n^k)$ is taken as a single variable.

\begin{figure*}[!t]
	{\small
\begin{eqnarray}\label{eq:augmented} 
&&\psi({\cal P},{\cal A},{\cal L},{\cal R},{\cal S},{\cal G}, {\cal D}, {\cal U}, {\cal Y})= \Gamma({\cal P},{\cal A}, {\cal L},{\cal R},{\cal S},{\cal G}, {\cal D}, {\cal U})\nonumber\\&&
+ \sum_ {k,n} y_{1,n}^k  \left(\sum_{t=1}^{k} l_n^t +\sum_{t=1}^{k}\sum_{m\in {\cal N}, m\ne n} r_{n,m}^t -\sum_{t=1}^{k}\sum_{m\in {\cal N}} r_{m,n}^t +\sum_{t=1}^{k}s_n^t+\sum_{t=1}^{k}D_n^t+u_{1,n}^k-E_n^k\right) \nonumber\\&&
+  \sum_ {k,n}y_{2,n}^k  \left(\sum_{t=1}^{k} l_n^t +\sum_{t=1}^{k}\sum_{m\in {\cal N}, m\ne n} r_{n,m}^t -\sum_{t=1}^{k}\sum_{m\in {\cal N}} r_{m,n}^t +\sum_{t=1}^{k}s_n^t+\sum_{t=1}^{k}D_n^t-u_{2,n}^k-E_n^k+B_n^{\max}\right) \nonumber\\&&
+  \sum_ {k,n}y_{3,n}^k  \left(p_n^k -l_n^k-s_{n}^k - g_n^k\right)+ \sum_ {k,n}y_{4,n}^k \left(p_n^k +u_{3,n}^k - P_n\right) + \sum_ {k,n}y_{5,n}^k \left(s_n^k +u_{4,n}^k - \sum_{m\in {\cal N}, m\ne n} r_{m,n}^k\right)\nonumber\\&&+\frac{\rho}{2}   \sum_ {k,n} \left( \sum_{t=1}^{k} l_n^t +\sum_{t=1}^{k}\sum_{m\in {\cal N}, m\ne n} r_{n,m}^t-\sum_{t=1}^{k}\sum_{m\in {\cal N}} r_{m,n}^t +\sum_{t=1}^{k}s_n^t +\sum_{t=1}^{k}D_n^t  +u_{1,n}^k-E_n^k\right)^2\nonumber\\&&
+ \frac{\rho}{2}   \sum_ {k,n}\left(\sum_{t=1}^{k} l_n^t +\sum_{t=1}^{k}\sum_{m\in {\cal N}, m\ne n} r_{n,m}^t -\sum_{t=1}^{k}\sum_{m\in {\cal N}} r_{m,n}^t +\sum_{t=1}^{k}s_n^t+\sum_{t=1}^{k}D_n^t-u_{2,n}^k -E_n^k+B_n^{\max} \right)^2 \nonumber\\&&
+ \frac{\rho}{2} \left(\sum_n a_n^k -1\right)^2
+  \frac{\rho}{2}  \sum_ {k,n} \left(p_n^k -l_n^k-s_{n}^k - g_n^k\right)^2+\frac{\rho}{2}  \sum_ {k,n} \left(p_n^k +u_{3,n}^k - P_n\right)^2\nonumber\\&&
+ \frac{\rho}{2}  \sum_ {k,n} \left(s_n^k +u_{4,n}^k - \sum_{m\in {\cal N}, m\ne n} r_{m,n}^k\right)^2+ \sum_ {k}y_{6}^k \left(\sum_n a_n^k -1\right).
\end{eqnarray}}
\vspace{-.2in}
\end{figure*}

%Let the augmented Lagrange $\Psi$ be defined as:\\
%\begin{multline}
%\psi({\cal P},{\cal L},{\cal R},{\cal G}, {\cal D}, {\cal U}, {\cal Y})= \Gamma({\cal P},{\cal L},{\cal R},{\cal G}, {\cal D}, {\cal U})\\
%+ \sum_ {k,n} y_{1,n}^k  (\sum_{t=1}^{k} l_n^t +\sum_{t=1}^{k}\sum_{m\in {\cal N}, m\ne n} r_{n,m}^t -\sum_{t=1}^{k}\sum_{m\in {\cal N}} r_{m,n}^t +\sum_{t=1}^{k}D_n^t+u_{1,n}^k-E_n^k) \\
%+  \sum_ {k,n}y_{2,n}^k  (\sum_{t=1}^{k} l_n^t +\sum_{t=1}^{k}\sum_{m\in {\cal N}, m\ne n} r_{n,m}^t -\sum_{t=1}^{k}\sum_{m\in {\cal N}} r_{m,n}^t +\sum_{t=1}^{k}D_n^t-u_{2,n}^k-E_n^k+B_n^{\max}) \\
%+  \sum_ {k,n}y_{3,n}^k  (p_n^k -l_n^k - g_n^k)+ \sum_ {k,n}y_{4,n}^k (p_n^k +u_{3,n}^k - P_n)\\+\frac{\rho}{2}   \sum_ {k,n} \parallel \sum_{t=1}^{k} l_n^t +\sum_{t=1}^{k}\sum_{m\in {\cal N}, m\ne n} r_{n,m}^t-\sum_{t=1}^{k}\sum_{m\in {\cal N}} r_{m,n}^t +\sum_{t=1}^{k}D_n^t+u_{1,n}^k-E_n^k\parallel_2^2\\
% + \frac{\rho}{2}   \sum_{t=1}^{k} l_n^t +\sum_{t=1}^{k}\sum_{m\in {\cal N}, m\ne n} r_{n,m}^t -\sum_{t=1}^{k}\sum_{m\in {\cal N}} r_{m,n}^t +\sum_{t=1}^{k}D_n^t-u_{2,n}^k-E_n^k+B_n^{\max} \parallel_2^2\\
% +  \frac{\rho}{2}  \sum_ {k,m} \parallel p_n^k -l_n^k - g_n^k\parallel_2^2+\frac{\rho}{2}  \sum_ {k,n} \parallel p_n^k +u_{3,n}^k - P_n\parallel_2^2
%\end{multline}

%Given the augmented Lagrange function above, at every iteration, ADMM  updates all primal and dual variables sequentially, which is summarized in Algorithm 3.

\begin{figure*}
 \begin{minipage}[h]{6.5 in}
\rule{\linewidth}{0.3mm}%\vspace{-.1in}

{\bf {\footnotesize Algorithm 1 -  Proximal Jacobian ADMM for Solving Proposed Problem  in \eqref{eq:oproblem}-\eqref{eq:ocst} }}\vspace{-.1in}\\
\rule{\linewidth}{0.2mm}\vspace{.1in}
{ {\small
\begin{tabular}{ll}
	1:  Initialization: $i=0$, $({\cal P},{\cal A},{\cal L},{\cal R},{\cal S},{\cal G}, {\cal D}, {\cal U}, {\cal Y})^{0}=(0,0,0,0,0,0,0,0,0)$\\
	 Specify the ADMM parameters $\rho$, $\tau$, and $\gamma$, and the convergence threshold $\eta$\\
%	&{\bf IF} $c_n=0$ {\bf THEN} Set $c_n=\beta$ {\bf ENDIF}\\
    2: ADMM Iteration: {\bf REPEAT}\\
	 \quad $(p_n^k,a_n^k)^{i+1} \leftarrow \argmin_{p_n^k,a_n^k} \psi(({\cal P,\cal A}\setminus p_n^k,a_n^k)^{i},p_n^k,a_n^k,{\cal L}^{i},{\cal R}^{i},{\cal S}^{i},{\cal G}^{i}, {\cal D}^{i}, {\cal U}^{i}, {\cal Y}^{i})$ $+\frac{1}{2}\tau (p_n^k-(p_n^k)^{i})^2+\frac{1}{2}\tau (a_n^k-(a_n^k)^{i})^2,%\text{ using \eqref{eq:q1}-\eqref{eq:q2}}, \quad
	  \forall$ $k$, and $n$ \\
	 \quad{\footnotesize$(l_n^k)^{i+1} \leftarrow\argmin_{l_n^k} \psi({\cal P}^{i},{\cal A}^{i},({\cal L}\setminus l_n^k)^{i},l_n^k,{\cal R}^{i},{\cal S}^{i},{\cal G}^{i}, {\cal D}^{i}, {\cal U}^{i}, {\cal Y}^{i})+\frac{1}{2}\tau(l_n^k-(l_n^k)^{i})^2,%\text{ using \eqref{lnkupdate}}, 
		\quad \forall$ $k$, and $n$}\\
	 \quad{\footnotesize$(r_{n,m}^k)^{i+1} \leftarrow\argmin_{r_{n,m}^k} \psi({\cal P}^{i},{\cal A}^{i},{\cal L}^{i},({\cal R}\setminus r_{n,m}^k)^{i},r_{n,m}^k, {\cal S}^{i},{\cal G}^{i}, {\cal D}^{i}, {\cal U}^{i}, {\cal Y}^{i})+\frac{1}{2}\tau (r_{m,n}^k-(r_{m,n}^k)^{i})^2$}, %\\	 \quad \quad \quad \text{ using \eqref{updatermn}},
	 $\forall$ $k$, $n$, and $m$, $n\ne m$\\
	 \quad{\footnotesize$(s_{n}^k)^{i+1} \leftarrow\argmin_{s_m^k} \psi({\cal P}^{i},{\cal A}^{i},{\cal L}^{i},{\cal R}^{i}, ({\cal S}\setminus s_n^k)^{i},s_n^k,{\cal G}^{i}, {\cal D}^{i}, {\cal U}^{i}, {\cal Y}^{i})+\frac{1}{2}\tau (s_n^k-(s_n^k)^{i})^2,%\text{ using \eqref{snkupdate}},
	 	 \ \forall$ $k$,  and $n$}\\
	 \quad{\footnotesize$(g_{n}^k)^{i+1} \leftarrow\argmin_{g_n^k} \psi({\cal P}^{i},{\cal A}^{i},{\cal L}^{i},{\cal R}^{i}, {\cal S}^{i}, ({\cal G}\setminus g_n^k)^{i},g_n^k,{\cal D}^{i}, {\cal U}^{i}, {\cal Y}^{i})+\frac{1}{2}\tau(g_n^k-(g_n^k)^{i})^2,%\text{ using \eqref{gnkupdate}},
	 	 \ \forall$ $k$,  and $n$}\\
	 \quad{\footnotesize$(d_{n}^k)^{i+1} \leftarrow\argmin_{d_n^k} \psi({\cal P}^{i},{\cal A}^{i},{\cal L}^{i},{\cal R}^{i}, {\cal S}^{i}, {\cal G}^{i},({\cal D}\setminus d_n^k)^{i},d_n^k, {\cal U}^{i}, {\cal Y}^{i})+\frac{1}{2}\tau(D_n^k-(D_n^k)^{i})^2,%\text{ using \eqref{updatednk}},
	 	 \ \forall$$k$ ,and $n$}\\
	 \quad$(u_{j,n}^k)^{i+1} \leftarrow\argmin_{u_{i,n}^k} \psi({\cal P}^{i},{\cal A}^{i},{\cal L}^{i},{\cal R}^{i}, {\cal S}^{i}, {\cal G}^{i},{\cal D}^{i},({\cal U}\setminus u_{j,n}^k)^{i},u_{j,n}^k,  {\cal Y}^{i})+\frac{1}{2}\tau(u_n^k-(u_n^k)^{i})^2$,%\\	 \quad \quad \text{ using \eqref{updateunk1}-\eqref{updateunk2}}, \quad 
	 \ $\forall$ $k$, $n$, and $j=1,2,3,4$\\
 \quad $(y_{1,n}^k)^{i+1} \leftarrow (y_{1,n}^k)^i+\gamma \rho(\sum_{t=1}^{k} l_n^t +\sum_{t=1}^{k}\sum_{m\in {\cal N}, m\ne n} r_{n,m}^t -\sum_{t=1}^{k}\sum_{m\in {\cal N}} r_{m,n}^t %$\\ \quad \quad$
  +\sum_{t=1}^{k}s_n^t+\sum_{t=1}^{k}D_n^t+u_{1,n}^k-E_n^k)^{i+1}\  \forall$ $k$,  and $n$\\		
 \quad $(y_{2,n}^k)^{i+1} \leftarrow (y_{2,n}^k)^i+\gamma \rho(\sum_{t=1}^{k} l_n^t +\sum_{t=1}^{k}\sum_{m\in {\cal N}, m\ne n} r_{n,m}^t -\sum_{t=1}^{k}\sum_{m\in {\cal N}} r_{m,n}^t$\\
 \quad \quad$  +\sum_{t=1}^{k}s_n^t+\sum_{t=1}^{k}D_n^t-u_{2,n}^k-E_n^k+B_n^{\max})^{i+1}\quad \forall$ $k$,  and $n$\\	
	\quad $(y_{3,n}^k)^{i+1} \leftarrow (y_{3,n}^k)^i+\gamma \rho(p_n^k -l_n^k-s_{n}^k - g_n^k)^{i+1}\quad \forall$ $k$,  and $n$\\
	 \quad $(y_{4,n}^k)^{i+1} \leftarrow (y_{4,n}^k)^i+\gamma \rho(p_n^k +u_{3,n}^k - P_n)^{i+1}\quad \forall$ $k$,  and $n$\\
 \quad $(y_{5,n}^k)^{i+1} \leftarrow (y_{5,n}^k)^i+\gamma \rho(s_n^k +u_{4,n}^k - \sum_{m\in {\cal N}, m\ne n} r_{m,n}^k)^{i+1}\quad \forall$ $k$,  and $n$\\
 \quad $(y_{6}^k)^{i+1} \leftarrow (y_{6}^k)^i+\gamma \rho(\sum_n a_n^k -1)^{i+1}\quad \forall$ $k$\\
 \quad $i\leftarrow i+1$\\
	 {\bf UNTIL} {\footnotesize$|\psi({\cal P}^{i},{\cal A}^{i},{\cal L}^{i},{\cal R}^{i},{\cal S}^{i},{\cal G}^{i}, {\cal D}^{i}, {\cal U}^{i}, {\cal Y}^{i}) - \psi({\cal P}^{i-1},{\cal A}^{i-1},{\cal L}^{i-1},{\cal R}^{i-1},{\cal S}^{i-1},{\cal G}^{i-1}, {\cal D}^{i-1}, {\cal U}^{i-1}, {\cal Y}^{i-1})|<\eta$}\\
\end{tabular}}}\vspace{.1in}\\
\rule{\linewidth}{0.3mm}
\end{minipage}
\vspace{-.2 in}
\end{figure*}
\begin{remark}
 Algorithm 1 is a distributed parallel algorithm. In particular, the variables associated with different $(n,k)$ can be  updated independently and in parallel.
\end{remark}

We note that there are seven sets of arg-mins in Algorithm 1. All these problems are convex problems (since the indicator functions are equivalent to linear constraints). The detailed solutions to these problems are given in Appendix, where the problems are solved in closed form, except the first where the solution is in terms of a solution to a single variable equation. %The equation number corresponding to each update is referred to in  Algorithm 1.
 The next theorem states the optimality of Algorithm 1.

\begin{theorem}
	Algorithm 1 optimally solves the problem in \eqref{eq:oproblem}-\eqref{eq:ocst}, and converges with an error rate of $o(1/b)$ after $b$ iterations when the ADMM parameters are chosen such that 
	\begin{equation}
	\tau >4K\rho \left(\frac{9NK+N^2K}{2-\gamma} -1\right),\label{condition}
	\end{equation}
	for $0<\gamma<2$, and $\rho>0$.
\end{theorem}
\begin{proof}
	The objective function $\Gamma({\cal P},{\cal A},{\cal L},{\cal R},{\cal S},{\cal G}, {\cal D}, {\cal U})$ is separable in all the variables $(p_n^k,a_n^k)$, $l_n^k$, $r_{m,n}^k$, $s_n^k$, $g_n^k$, $D_n^k$, and $u_n^k$ for all $m, n, $ and $k$, and the constraints are linear equalities. The problem is convex optimization and all the separable functions are closed proper convex, which satisfies the  assumptions for the optimality of Proximal Jacobian ADMM in \cite{2013arXiv1312.3040D}. % are both satisfied. The first assumption is that all the separable functions are closed proper convex, and the second assumption is that there is a saddle point to the problem. 
	Further, the choice of parameters in (\ref{condition}) satisfy the parameter conditions in \cite{2013arXiv1312.3040D} for the $9NK+N^2K$ number of variables ( $(p_n^k,a_n^k)$, $l_n^k$, $r_{m,n}^k$, $s_n^k$ $g_n^k$, $D_n^k$, and $u_n^k$ for all $m, n, $ and $k$ ), and that each variable is in at-most $4K$ constraints with absolute multiplicative coefficient of $1$ in the problem (\ref{eq:oproblem})-(\ref{eq:ocst}).
	%
	%Appendix A shows that all the assumptions for optimality in \cite{2012arXiv1208.3922H} are satisfied, and thus the optimality and the linear convergence properties of ADMM hold for $\alpha$ sufficiently small.
\end{proof}

The next result gives the computational complexity of the proposed algorithm. 

\begin{theorem}
	Each iteration of Algorithm 1 has $O(N^2K^2)$ computational complexity. 
\end{theorem}
\begin{proof}
We note that the detailed steps the sub-problems are given in the Appendix. Problem 1 for each $n$ and $k$ involves solving a single-variable equation and is thus  $O(NK)$ computational complexity. Problem 2 for each $n$ and $k$ needs a sum over $v$ from $k$ to $K$ and thus has  $O(NK^2)$  complexity. Problem 3 is for energy cooperation which for every $m$, $n$, and $k$ requires $O(K)$ time, and has $O(N^2K^2)$  complexity. We assume that for each $n$ and $k$, $\sum_{m\ne n} r_{m,n}^k$ can be computed and stored which is $O(N^2K)$ complexity. Using this, Problem 4 has $O(NK^2)$  complexity. Similarly, Problem 5 has $O(NK)$  complexity. Problem 6 has $O(NK^2)$  complexity. Having stored values of $\sum_{m\ne n} r_{m,n}^k$ and $\sum_{m\ne n} r_{n,m}^k$ for each $n$ and $k$, Problem 7 has a complexity of $O(NK^2)$. Thus, the overall complexity is dominated by the complexity of Problem 3, and is $O(N^2K^2)$.
\end{proof}

\begin{remark}
	Without any energy cooperation, $r_{m,n}^k=0$ and the overall complexity of the proposed algorithm is $O(NK^2)$. 
\end{remark}

We further note that this problem was considered without the energy cooperation and grid energy in \cite{ZheJSAC16} and the proposed algorithm in \cite{ZheJSAC16} has the same computational complexity of  $O(NK^2)$. However, the algorithm in \cite{ZheJSAC16} performs the first step of calculating the optimal discharge allocation. However, in the presence of grid energy and cooperation, such allocation cannot be found since the energy can be transferred to other nodes rather than discharging. Having removed the discharge variables, there were only two variables for power and bandwidth left, which could be solved using an alternating minimization based approach in \cite{ZheJSAC16}. However, we have many more sets of variables. In order to use Proximal Jacobian ADMM, we had to use the power and bandwidth as the variables in a single block. Rather than performing an alternating minimization over these variables, the joint optimization is solved in this paper. Thus, this paper gives an alternate way of solving the algorithm in \cite{ZheJSAC16} with the same complexity. In addition, the proposed algorithm in this paper works for the energy cooperation and grid energy parameters while the approach of \cite{ZheJSAC16} do not extend easily.

We finally note that the cooperation may involve a loss due to efficiency of transfer which can be easily incorporated in the given constraints without changing the problem structure or the approach. In order to keep the expressions simpler, we have not included the efficiencies for energy transfer as well as battery charging or discharging.

\vspace{-.2in}

\section{Simulation and Results}

%First of all we want to evaluate the energy cost trade-off, so we fix the bandwidth allocation, and we assume equal channel gain and equal weight for all users. We consider five broadcast channels $N=5$, and each one communicate with one receiver ${\cal M}_1=1$, ${\cal M}_2=1$,....., ${\cal M}_{5}=1$. We set the scheduling period as K=10 slots. For each transmitter n, we set the initial battery level $B_n^0=0$ and the battery capacity $B_n^{max}=20$ units. We assume the harvested energy $E_n^k$ follows a truncated Gaussian distribution with mean $\mu_n=10$ and variance of 2, ${\cal N}(10,2)$. We first assume only harvested and grid energy (no -energy cooperation).We vary $\lambda$ (the grid energy use cost factor), and we plot the achievable system throughput-cost curve. Note that as $\lambda$ increases, the system tends to rely less on the grid energy and more on the harvested one. For high value of $\lambda$, the system uses only the harvested energy, so it is expected to have lower system throughput but at the same time lower cost. In contrast to that, if $\lambda=0$, the system can use as much as possible of grid energy to serve the users in the highest achievable rate constrained to the maximum power, but the cost is very high in this case.

In this section, we will evaluate the proposed algorithm in different scenarios.

\vspace{-.15in}
\subsection{Impact of Grid Energy Cost}\label{sec:conv}

%\begin{figure}[ht]
%\vspace{-2mm}
%\begin{center}
%\includegraphics[trim = 0in 2.3in 0in 2.5in, clip, width= 0.45\textwidth]{./figs/new_figures/figure_sys_rate_iter_4500.pdf}
%%\vspace{-6mm}
%\caption{\small Decrease in system throughput with increasing $\lambda$ for different values of $\Delta$.}
%\label{fig:rate_grid}
%\end{center}
%%\vspace{-.3in}
%\end{figure}

\begin{figure*}[ht]
\vspace{-3mm}
\begin{minipage}{0.31\textwidth}
\begin{center}
\includegraphics[trim = 0in 2.3in 0in 2.5in, clip, width= \textwidth]{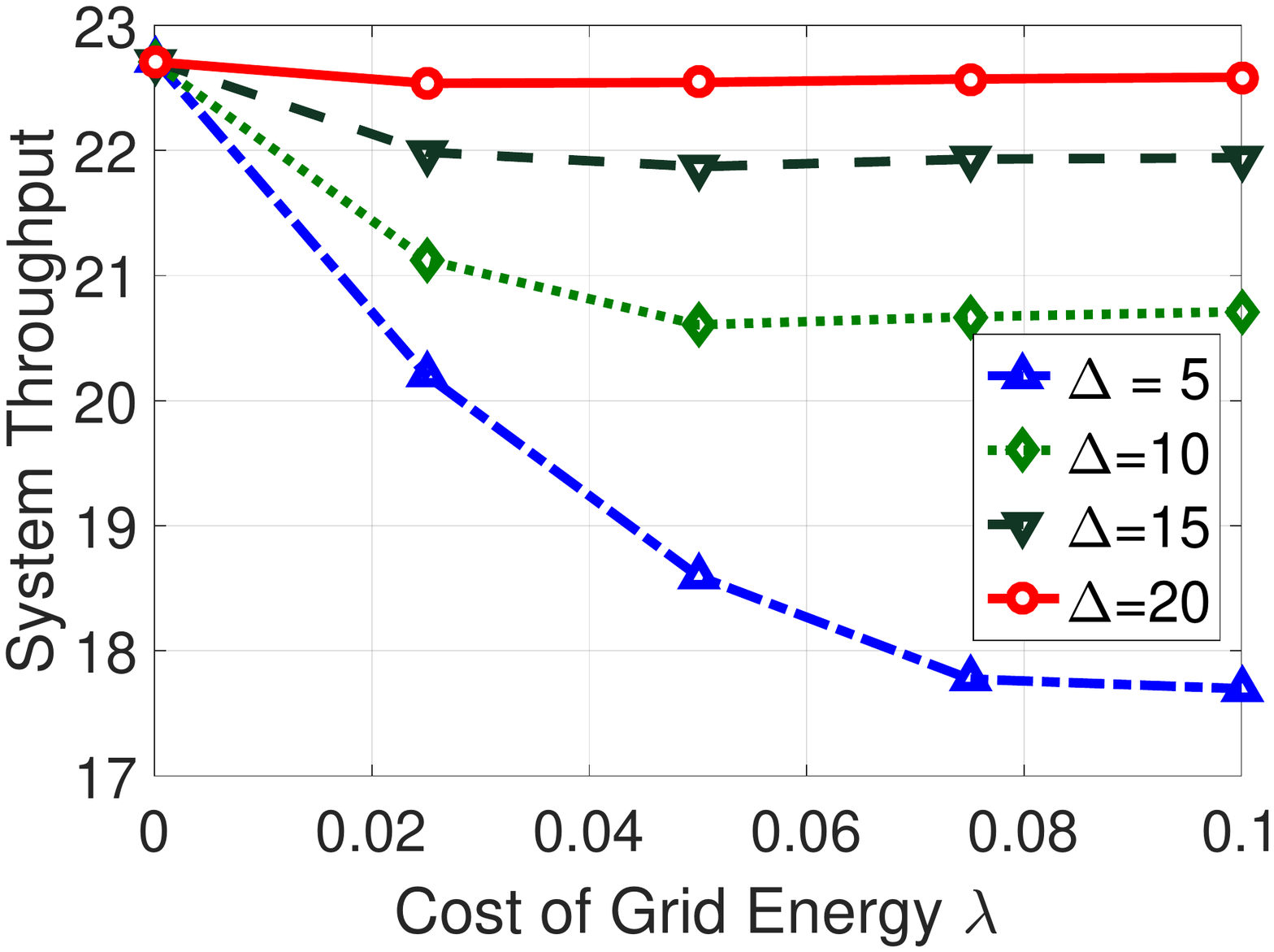}
%\vspace{-6mm}
\caption{\small Decrease in system throughput with increasing $\lambda$ for different values of $\Delta$.}
\label{fig:rate_grid}
\end{center}
\end{minipage}
\hspace{.02\textwidth}
\begin{minipage}{0.31\textwidth}
	\begin{center}
		{\includegraphics[trim = 0.2in 2.3in 0.5in 2.5in, clip, width=.98\textwidth,draft=false]{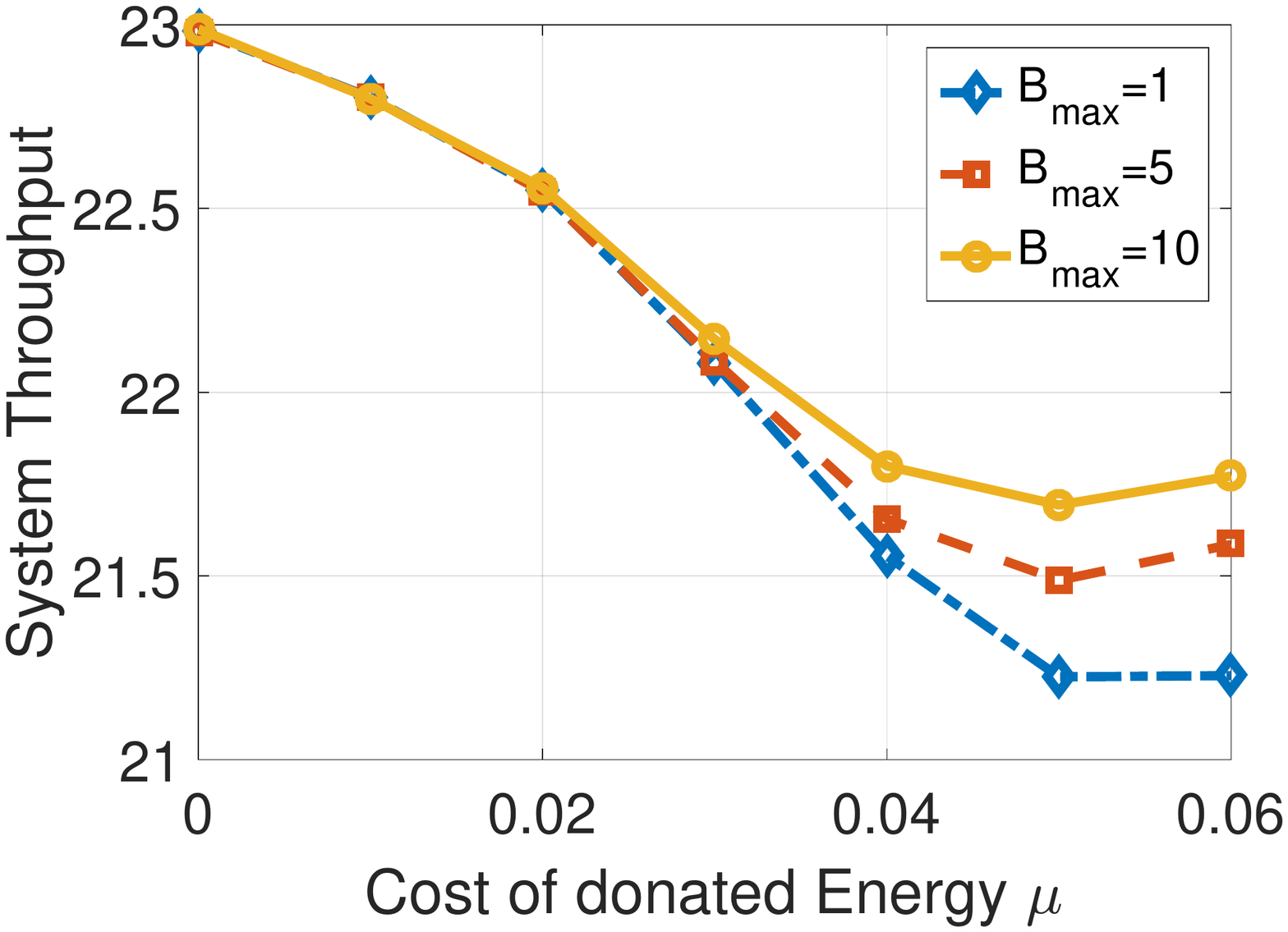}}\vspace{-.2in}
		\caption{\small Decrease of rate with increase in cost of donation, $\mu$, for different values of battery capacity $B_{\max}$.  }
		\label{fig:rate_donation}
	\end{center}
\end{minipage}
\hspace{.02\textwidth}
%\vspace{3mm}
\begin{minipage}{0.31\textwidth}
	\begin{center}
		{\includegraphics[trim = .2in 2.3in .5in 2.5in, clip, width=\textwidth,draft=false]{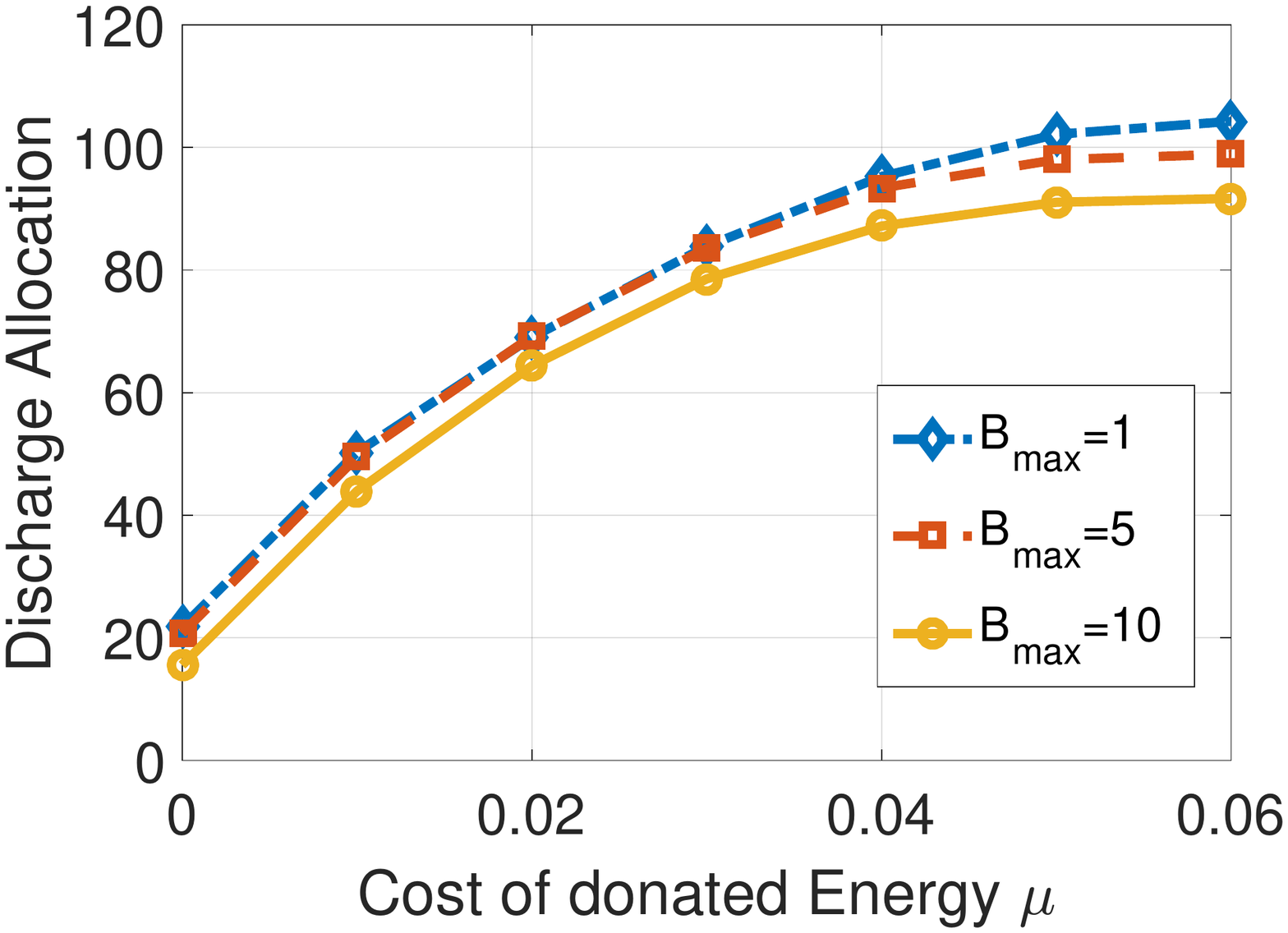}}\vspace{-.2in}
		\caption{\small Increase of discharged energy with increase in cost of donation, $\mu$, for different values of battery capacity $B_{\max}$.}
		\label{fig:discharge_donation}
	\end{center}
\end{minipage}
\vspace{-.1in}
\end{figure*}

We consider $N=5$ users with a scheduling period of $K=5$ time slots. The weights of all the users are taken to be identical,  $W_n=1$.  We assume that the maximum power constraints is $P_n=20W$, and maximum battery capacity is $B_n^{\max}=20W$. The channel gains $\{h_n^k\}$ are distributed as ${\sf CN}(0,1)$. For the energy arrival $E_n^k-E_{n}^{k-1}$, a truncated Gaussian distribution is used which is given by the maximum of zero and a Gaussian random variable with mean $\Delta_n$ and variance $4$. Let $\Delta_n = \Delta$, independent of $n$. The Proximal Jacobian ADMM parameters are $\rho=10^{-3}, \gamma=1$, $\tau =0.5$, and the convergence threshold for the iterative loop is chosen to be $\eta = 10^{-6}$. Let the cooperation cost to be $\mu=0.2$. The grid energy cost $\lambda$ is chosen as a variable.

%We choose $\Delta_n = \frac{200n}{N(N+1)}$, for $n= 1, \cdots, N$, such that $\sum_n \Delta_n = 100$,, i.e., fixed total mean harvested energy irrespective of the value of $N$. The grid energy cost $\lambda$, and the energy cooperation cost $\mu$ are both chosen to be equal to 0.01. The ADMM parameters are $\rho=1, \alpha=0.1$, and the convergence threshold for the inner ADMM loop is chosen to be $\eta = 10^{-3}$.

%We consider $N=20$ users, and other system parameters are the same as those in Sect.  \ref{sec:conv} except  for $\mu=0.2$,  $\Delta_n = \Delta$, independent of $n$, and we vary the grid energy cost $\lambda$. 

The results are averaged over 12 runs with different realizations of  channel gains and energy arrivals. Figure \ref{fig:rate_grid} shows the decrease of system throughput (in nats) with increasing  grid energy cost,  $\lambda$, for four different values of $\Delta = 5, 10, 15, $ and $20$J. We note that when $\lambda=0$, the system throughput is independent of $\Delta$ since the energy from grid can be taken up to the maximum power constraint in the chosen bandwidth. Thus, the problem in this case becomes  a bandwidth allocation problem with each node using the maximum energy constraint in each time slot. Due to similar channel gains for all links, assuming equal bandwidth for all links, and ignoring the random effect of channel gains give the system rate as $5\times \log(1+ 20\times 5) \approx 23$ nats. Thus, the result at $\lambda=0$ is almost equal to this. As the use of grid energy becomes more expensive, the system throughput reduces. For $\Delta=20$J, more energy arrives at each node and thus there is little decrease in throughput with increasing $\lambda$ as compared to the case where grid energy is free. However, for smaller value of $\Delta$, the system throughput decreases significantly with $\lambda$. An optimal operating point can be chosen based on the system design requirements.

%\begin{figure*}[ht]
%\vspace{-3mm}
%\begin{minipage}{0.45\textwidth}
%\begin{center}
%{\includegraphics[trim = 0.2in 2.3in 0.5in 2.5in, clip, width=\textwidth,draft=false]{./figs/new_figures/sys_rate.pdf}}\vspace{-.2in}
%\caption{\small Decrease of rate with increase in cost of donation, $\mu$, for different values of battery capacity $B_{\max}$.  }
%\label{fig:rate_donation}
%\end{center}
%\end{minipage}
%\hspace{.05\textwidth}
%%\vspace{3mm}
%\begin{minipage}{0.45\textwidth}
%\begin{center}
%{\includegraphics[trim = .2in 2.3in .5in 2.5in, clip, width=\textwidth,draft=false]{./figs/new_figures/sys_discharge.pdf}}\vspace{-.2in}
%\caption{\small Increase of discharged energy with increase in cost of donation, $\mu$, for different values of battery capacity $B_{\max}$.}
%\label{fig:discharge_donation}
%\end{center}
%\end{minipage}
%%\vspace{-.25in}
%\end{figure*}
\vspace{-.15in}
\subsection{Impact of Energy Cooperation Cost}	\label{sec:simulation_cooperation}

We consider $N=5$ users with $\Delta_n = 5n$J, $n=1, \cdots, 5$.  The maximum battery capacity of each node is chosen to be the same $B_n^{\max} = B_{\max}$. The grid energy cost $\lambda= 0.1$, and $\mu$ is \ a variable.  All other parameters are same as those in Sect. \ref{sec:conv}. Figure \ref{fig:rate_donation} plots the system throughput with respect to $\mu$ for different values of $B_{\max}$. We note that the system throughput decreases with $\mu$ since the donation across nodes is more expensive. Figure \ref{fig:discharge_donation} demonstrates that the amount of energy that is discharged and thus not used also increases with the increase in $\mu$. This is because  different nodes have different average incoming energy and by penalizing donation, the energy does not get evenly distributed.

\begin{figure*}[ht]
\vspace{-3mm}
\begin{minipage}{0.31\textwidth}
\begin{center}
{\includegraphics[trim = 0.3in 2in 1.2in 2.5in, clip, width=\textwidth,draft=false]{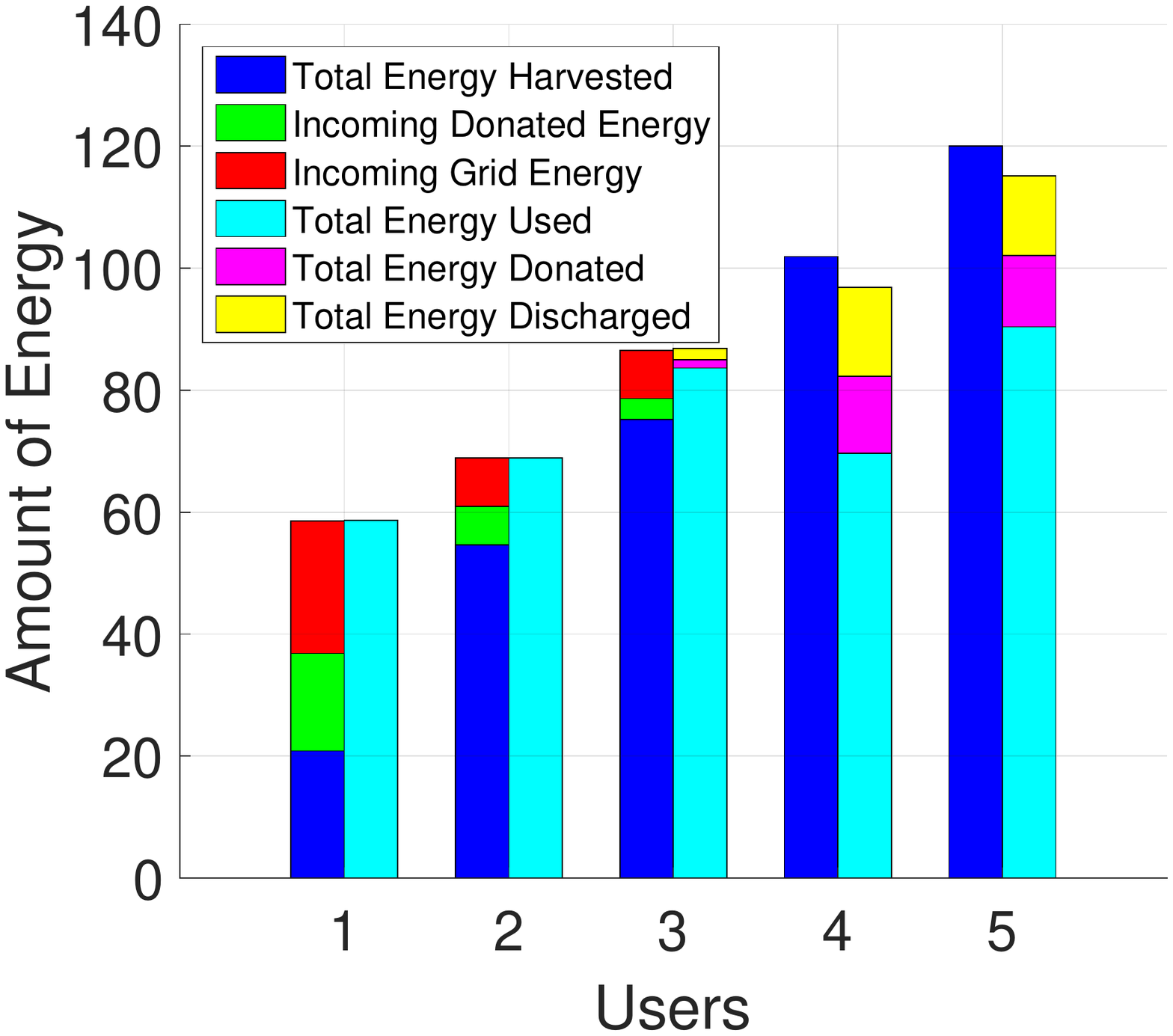}}\vspace{-.1in}
\caption{\small Incoming and outgoing energy for different nodes, $\mu=0.01$.  }
\label{fig:diffen_mu}
\end{center}
\end{minipage}
\hspace{.02\textwidth}
%\vspace{3mm}
\begin{minipage}{0.31\textwidth}
\begin{center}
{\includegraphics[trim = 0.3in 2in 1.2in 2.5in, clip, width=\textwidth,draft=false]{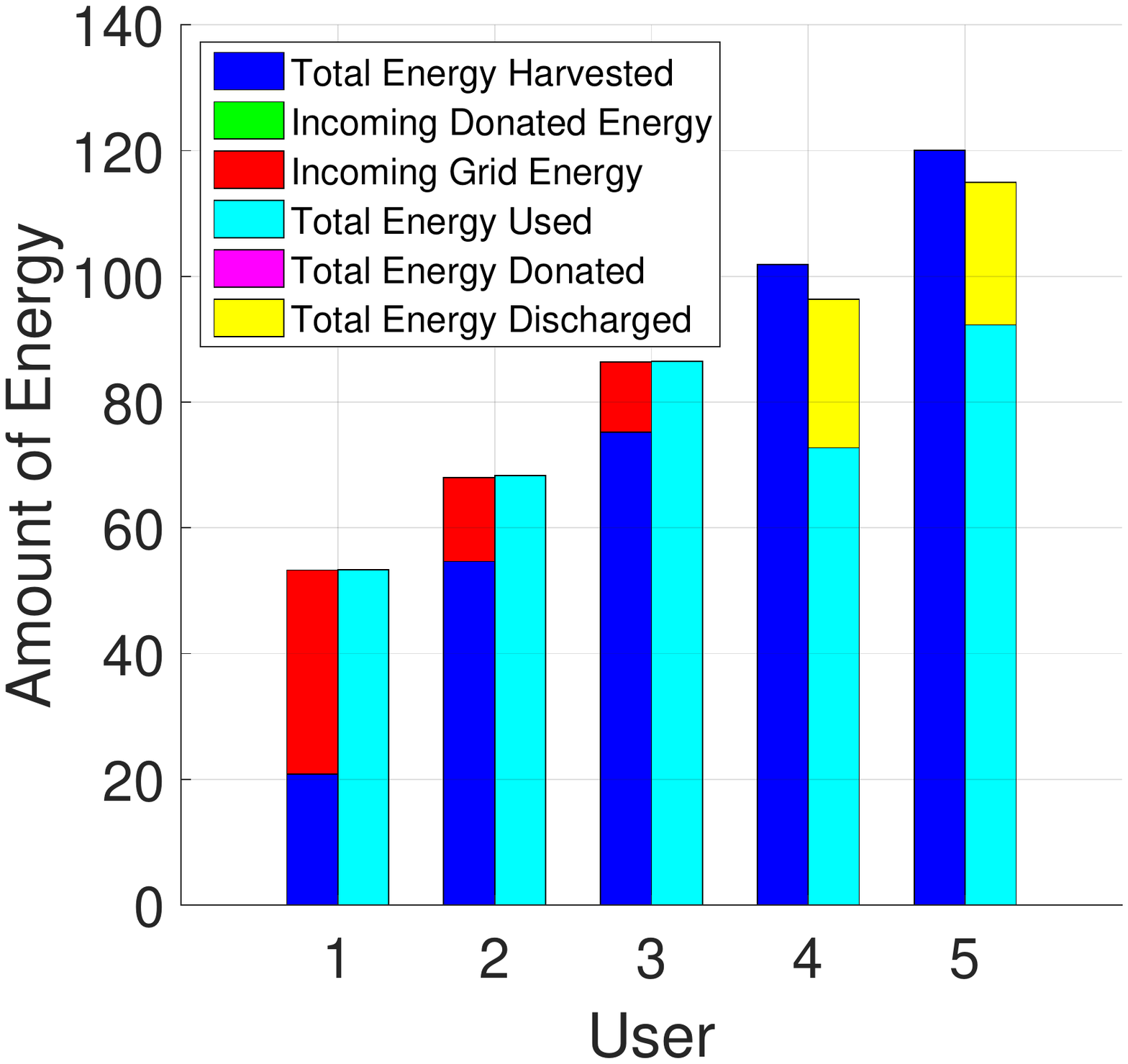}}\vspace{-.1in}
\caption{\small Incoming and outgoing energy for different nodes, $\mu=0.1$. }
\label{fig:diffen_muin}
\end{center}
\end{minipage}
\hspace{.02\textwidth}
%\vspace{3mm}
\begin{minipage}{0.31\textwidth}
	\begin{center}
		{\includegraphics[trim = 0.2in 2in 0.3in 2.5in, clip, width=\textwidth,draft=false]{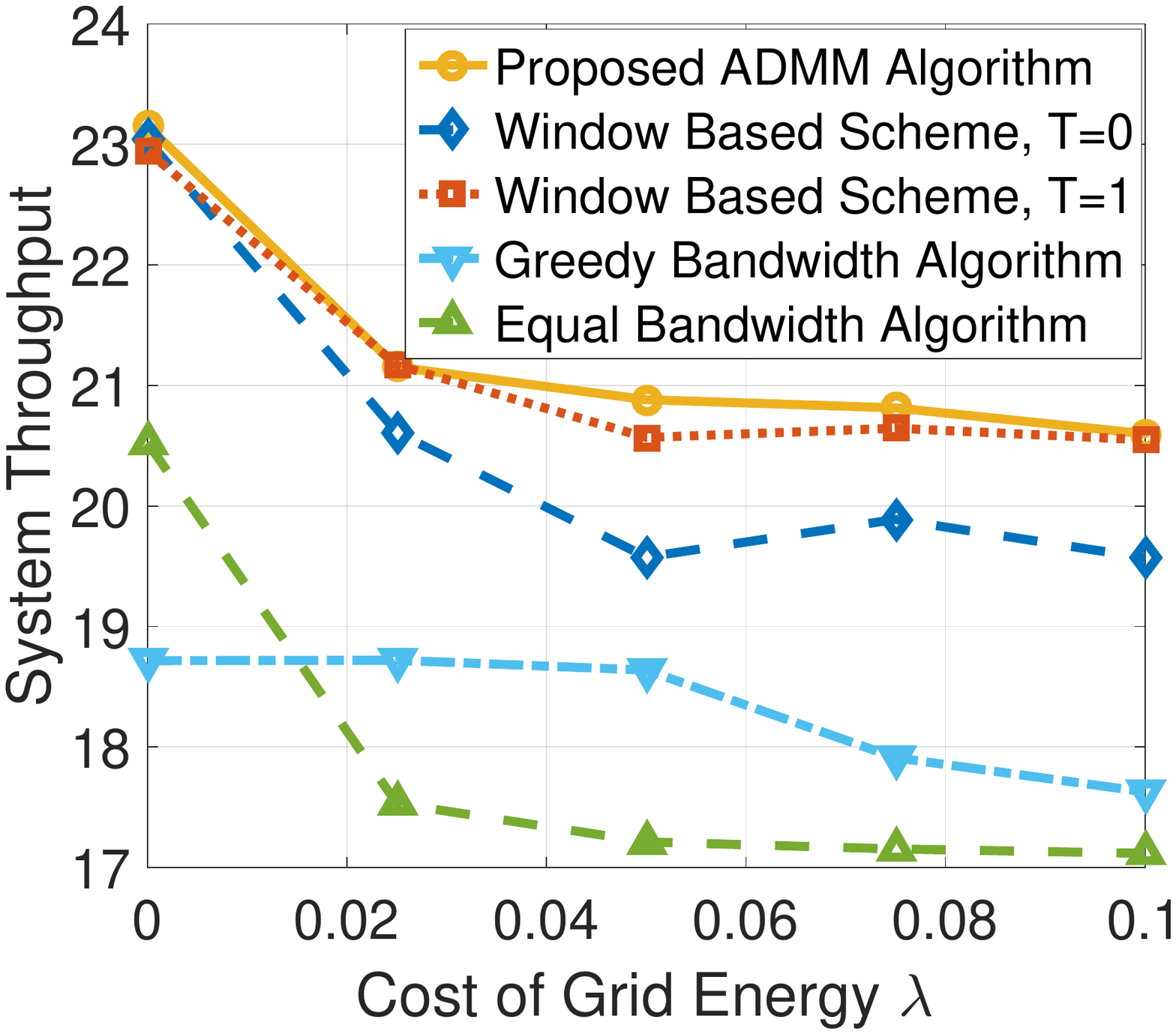}}\vspace{-.2in}
		\caption{\small System throughput  of proposed ADMM algorithm, as compared with the  different baseline schemes.  }
		\label{fig:rate_baseline}
	\end{center}
\end{minipage}
\vspace{-.25in}
\end{figure*}

\vspace{-.15in}
\subsection{Energy Arrival and Utilization }
We will now compare the amounts of energy that enters a node, and that leaves a node including the power consumption for communication transmission.
% We consider $N=5$ users scheduled over $K=5$ time slots. The means of the energy arrival processes are set as $\Delta_n = 3(n-1)$, for $n=1, \cdots, 8$.  
The grid energy cost $\lambda=0.01$, $B_n^{\max}=10W$, and all other parameters  are the same as those in Section \ref{sec:simulation_cooperation}. For a given realization of channel gains and energy arrivals, we find the total amount of incoming energy at each node and separate it  into the amount of energy harvested ${\cal L}$, the amount of energy donated by other nodes ${\cal R}$ (incoming part), and the amount of grid energy ${\cal G}$. We also find the amount of energy that leaves the node and separate it into  the amount of energy used for communication transmission ${\cal P}$, the amount of energy donated to other nodes ${\cal R}$ (outgoing part), and the amount of energy discharged ${\cal D}$. We consider two  donation cost values, $\mu = 0.01, 0.1$, in Figures \ref{fig:diffen_mu} and \ref{fig:diffen_muin} respectively.  The first three nodes (with lower harvested energy arrival) have incoming donated energy, while the nodes with higher harvested energy have outgoing energy when $\mu$ is small. This is because the later nodes can store energy in the first nodes to avoid going over the battery which can be transferred back in  future time slots, and this impact is higher at smaller values of $\mu$. It is easier to donate energy at lower value of $\mu$ and thus external grid energy used is lower, and there is less discharge as compared to higher value of $\mu$. At lower $\mu$, more energy can be used by the users with less harvested energy.  Finally, we note that the total amount of outgoing energy is slightly less than the total amount of incoming energy in some nodes. This is because a residual amount of energy is stored in the battery. 
%Thus, smaller value of $\mu$ encourages efficient energy cooperation among the nodes.

%\begin{figure}[ht]
%	\vspace{-3mm}
%%	\begin{minipage}{0.45\textwidth}
%		\begin{center}
%			{\includegraphics[trim = 0.2in 2in 0.3in 2.5in, clip, width=.45\textwidth,draft=false]{./figs/new_figures/sys_rate_baseline.pdf}}\vspace{-.2in}
%			\caption{\small System throughput  of proposed ADMM algorithm, as compared with the  different baseline schemes.  }
%			\label{fig:rate_baseline}
%		\end{center}
%	\end{minipage}
%	\hspace{.05\textwidth}
	%\vspace{3mm}
%	\begin{minipage}{0.45\textwidth}
%		\begin{center}
%			{\includegraphics[trim = .2in 2in .5in 2.5in, clip, width=\textwidth,draft=false]{./figs/new_figures/sys_obj_baseline.pdf}}\vspace{-.2in}
%			\caption{\small System objective comparison of proposed ADMM algorithm with different baseline schemes. }
%			\label{fig:obj_baseline}
%		\end{center}
%	\end{minipage}
	%\vspace{-.25in}
%\end{figure}

\vspace{-.15in}
\subsection {Comparison with  Baseline Schemes}\label{sec:baselines}

We consider three baseline schemes. The first is a window-based scheme where at each time,  each node only has the information of the energy arrival in its current and the next $T<K$ number of time slots. At each time $i$, we perform the optimization for $\min(T+1,K-i+1)$ time-slots using the proposed Proximal Jacobian ADMM strategy and use these decisions for the current time slot. { Even though the complexity of the algorithm at  each time is lower ($O(N^2 (T+1)^2)$), the process is repeated at each time (thus having overall complexity of $O(N^2T^2K)$).} The second is equal bandwidth strategy, where the bandwidth is equally divided among all transmission links, and the remaining variables are optimized using the Proximal Jacobian ADMM algorithm. The third is a greedy bandwidth scheme, which  assigns all available bandwidth to the link with the highest channel gain. Given such bandwidth allocation, the rest of the variables are optimized based on the proposed Proximal Jacobian ADMM strategy.

\begin{table}
	\vspace{-.1in}
	\caption{\small Convergence speed for varying $N$.}	\label{table:convergence}
	\centering
		\vspace{-.1in}
\resizebox{.45\textwidth}{!}{
%
%	\centering
	\begin{tabular}{| c | c | c | c |}
		\hline
		$N$ & Iterations to converge & Time / iteration (ms) & System throughput\\ \hline
		5 & 3715 & 88.5 & 21.4 \\ \hline
		10 & 4004 & 189.7 & 25.2\\ \hline
		15 & 4008 & 299.7 & 28.1\\ \hline
		20 & 4050 & 426.6 & 29.8\\ \hline
		25 & 4014 & 573.7 & 31.0\\ \hline
		30 & 4085 & 733.8 & 31.9\\ \hline
	\end{tabular}}
	\vspace{-.3in}
\end{table}

%Now we discuss the comparison of the proposed ADMM algorithm with two categories of baseline schemes: the window based scheme and the bandwidth allocation scheme. For the window based schemes, rather than the full $K$ number of time slots, only a window period of $T<K$ time slots is taken into account. That is, at a given time slot, the system only has the information of the energy arrival in its current and the next $T$ number of time slots. The system then carries out the proposed ADMM algorithm within such limited time window. It is a sequential scheme such that all the time slots are covered in at least one time window. The bandwidth baseline schemes considered are the equal bandwidth strategy, which evenly divides the total available bandwidth to each pair of transmission link, and the greedy bandwidth scheme, which namely assigns all available bandwidth to the link with the highest gain. 

We consider $N=5$ users,  $K=5$ time units,  $\Delta_n=\Delta=10$, and variance of energy arrival in each time slot as $36$. Let $\mu=0.8$ and all other parameters be the same as those in Section \ref{sec:conv}. Figure \ref{fig:rate_baseline} depicts the weighted sum rate (or the system throughput) for the proposed algorithm as compared to the different base-lines for varying $\lambda$.  For window-based schemes, we consider $T=0$ and $T=1$. As $T$ increases, the performance improves, and will reach optimal when $T=K-1$. When $T=0$, the algorithm is causal since it does not use any future information. However, the algorithm does not try to save for future ({ since the optimization is performed only for the current time-slot}). { Thus, there is a bigger difference between $T=0$ and $T=1$, while there is diminishing return since there is not a big difference between $T=1$ and $T=4$ (offline scheme). We note that the closeness between $T=1$ and $T=4$ are for the particular parameters chosen, while the diminishing returns in $T$ should hold in general. } Thus, with limited prediction, window-based schemes can potentially be used where in each window, the proposed algorithm is used for system optimization. We also note that the proposed strategy significantly outperforms the greedy and equal bandwidth strategies thus depicting that the joint optimization over all variables is necessary. 

%The results of the system throughputs and the system objectives are presented in Figure \ref{fig:rate_baseline} and Figure \ref{fig:obj_baseline} respectively, for two window based schemes with $T=0$ and $T=1$, as well as the bandwidth baseline schemes. The proposed ADMM algorithm achieves the highest rate and system objective value, and it is a special case of the window based scheme with $T=K$. The figures show that the window based scheme with higher $T$ achieves greater value of throughput. An explanation is the larger window of $T$ gives the system a better understanding of the amount of energy harvested in longer future, so that the system can plan and spend the energy more efficiently in the time slots with higher gains. In the extreme case of $T=0$, the system has no knowledge of the upcoming harvested energy, so every node will greedily optimize only for the currently time slot, and it is suboptimal comparing to the proposed ADMM algorithm. On the other hand, the proposed ADMM algorithm outperforms the two bandwidth baselines. Those baselines essentially have lower dimensions since the bandwidth variables are fixed, while they are jointly optimized with the power variables in the proposed ADMM algorithm, and hence they achieve lower values.

\vspace{-.15in}
\subsection{Convergence with Increasing Number of Users}

We consider increasing the number of users $N$, for the parameters as in Section \ref{sec:baselines}, and $\lambda=0.01$. The results are presented in  Table \ref{table:convergence}. The test was carried out on a 64-bit desktop with 3.5 GHz quadcore processor and 20 GB RAM. We see that the number of iterations to converge to a value lower than the chosen threshold $\eta = 10^{-6}$ generally increases with $N$, since the number of variables is growing in the expanding system. On the other hand, the time taken per iteration increases as well. At these parameters, the complexity seems to increase slightly more than linearly with $N$ even though theoretically, the complexity scales as $N^2$. This is because the only step that takes $N^2K^2$ complexity is the update of variables $r_{n,m}^k$, which potentially do not take that dominant time at the considered parameters, and all other updates are $O(NK^2)$. The system throughput  increases as the number of users increase. However, we see diminishing gains with increasing number of users. { We  note that the parallelization of the algorithm for different $(n,k)$ can significantly reduce the time per iteration in practice. } % the increment diminishes since the total available bandwidth is fixed for all cases.

%From the result we show that the numbers of iterations to converge are almost the same for all values of $N$, which take about 4000 iterations. The computation time is approximately linear with $N$, which is expected since the number of computational operations increases approximately linearly with $N$.

\vspace{-.2in}

\section{Conclusions}

We have treated the energy-bandwidth allocation problem for multiuser network  where each node is powered with both renewable and grid energy, nodes can cooperate, and each node has a limited battery capacity and finite transmission power. The objective is to maximize the weighted sum throughput, and minimize the use of grid energy and the amount of energy cooperation. An iterative algorithm based on the Proximal Jacobian ADMM is proposed and proved to be optimal. Numerical results demonstrate the different tradeoffs in the optimal solution. { Extension of the results with a data buffer at the transmitter is an open problem.  Using stochastic information of the energy arrivals and the channel gains to come up with optimal online algorithms is an interesting future direction. Incentive based mechanisms to help increase users' willingness to donate energy is left for the future.}

%We note that other proposed variants for multi-variable ADMM like \cite{2014arXiv1406.4064W,2015arXiv150204391R,2013arXiv1312.3040D} may also be used instead of the ADMM proposed in this paper, and exploring the best variant remains further work.
\vspace{-.2in}

\bibliographystyle{IEEEtran}
\bibliography{energyBwBib,ourenergypapers}
\newpage
\clearpage
\section{Appendix: Solving Seven Optimization Problems}

For primal updates, we have seven sets of problems. We now consider solving these problems one by one. Note that we will ignore the iteration numbers, accounting that the last values of the other variables are used.

\noindent {\bf Problem 1: Updating $({\cal P},{\cal A})$ }

\begin{eqnarray}
&&(p_n^k,a_n^k)=\argmin_{p_n^k \geq 0,a_n^k \geq 0} \left(-W_n  a_n^k  \log(1+p_n^kH_n^k/a_n^k)  \right. \nonumber\\&& + y_{3,n}^kp_n^k + y_{4,n}^kp_n^k+y_{6,n}^ka_n^k + \frac{\rho}{2} \left(p_n^k - l_n^k-s_n^k- g_n^k\right)^2  \nonumber \\ &&+ \frac{\rho}{2} \left(p_n^k-P_n +  u_{3,n}^k\right)^2+\frac{\rho}{2}(a_n^k+\sum_{j,j\neq n}^N a_j^k -1)^2\nonumber\\&&\left. +\frac{1}{2}\tau(p_n^k-(p_n^k)^{i})^2+\frac{1}{2}\tau(a_n^k-(a_n^k)^{i})^2\right).
\end{eqnarray}

This optimization is for a jointly convex function that is not differentiable at $a_n^k=0$ at which $p_n^k = 0$. The KKT conditions for $a_n^k>0$ are as follows. 

\begin{eqnarray}
 &&-\frac{W_n  H_n^k}{1+p_n^kH_n^k/a_n^k} + y_{3,n}^k + y_{4,n}^k+  {\rho} \left(p_n^k - l_n^k-s_n^k- g_n^k\right) \nonumber\\&& + {\rho} \left(p_n^k-P_n +  u_{3,n}^k\right)+\tau(p_n^k-(p_n^k)^{i}) = 0 \label{eq:q1}, \\
 &&-W_n   \log(1+p_n^kH_n^k/a_n^k)  +\frac{W_n p_n^kH_n^k/a_n^k}{  1+p_n^kH_n^k/a_n^k} +y_{6,n}^k\nonumber\\&& +{\rho}(a_n^k+\sum_{j,j\neq n}^N a_j^k -1)+\tau(a_n^k-(a_n^k)^{i})=0 \label{eq:q2}.
\end{eqnarray}
From \eqref{eq:q1}, we can solve $a_n^k$ as $a_n^k = p_n^k H_n^k/$  $\left(\frac{W_n H_n^k}{ y_{3,n}^k + y_{4,n}^k+  {\rho} \left(p_n^k - l_n^k-s_n^k- g_n^k\right) + {\rho} \left(p_n^k-P_n +  u_{3,n}^k\right)+\tau(p_n^k-(p_n^k)^{i})} -1\right)$, which can be substituted in \eqref{eq:q2} where we get an equation with a single variable which can be solved. If there is a solution $(p_n^k,a_n^k)$ with $p_n^k\ge 0$, and $y_{3,n}^k + y_{4,n}^k+  {\rho} \left(p_n^k - l_n^k-s_n^k- g_n^k\right) + {\rho} \left(p_n^k-P_n +  u_{3,n}^k\right)+\tau(p_n^k-(p_n^k)^{i})<W_n H_n^k$, this is the required solution. Else, $a_n^k = p_n^k =0$.

%Expanding the square terms and canceling all terms that are not function of $p_n^k$ yields,
%\begin{eqnarray}
%p_n^k&=& \argmin_{p_n^k \geq 0} \left(-W_n  a_n^k  \log(1+p_n^kH_n^k/a_n^k) +(y_{3,n}^k+ y_{4,n}^k)p_n^k \right. \nonumber \\ && \left.+\frac{\rho}{2} \left((p_n^k)^2 - 2p_n^kl_n^k-2p_n^ks_n^k- 2p_n^kg_n^k\right) + \frac{\rho}{2} \left((p_n^k)^2-2p_n^k P_n +2p_n^k  u_{3,n}^k\right)\right)
%\end{eqnarray}
%The argument is strictly convex with respect to $p_n^k$. Taking derivative to zero, and solving for $p_n^k$ gives $c_1(p_n^k)^2+c_2p_n^k+c_3=0$, where
%$c_1= 2\rho H_n^k/a_n^k$, $c_2=\rho(2\rho+H_n^k/a_n^k(y_{3,n}^k+y_{4,n}^k+\rho(u_{3,n}^k-l_n^k-s_n^k - g_n^k-P_n)))$, and
%$c_3=\rho(u_{3,n}^k-l_n^k-s_n^k - g_n^k-P_n)+y_{3,n}^k+y_{4,n}^k-W_n H_n^k$. Thus, if $c_2^2 <4c_1c_3$, the derivative is either always positive in which case optimal answer is zero, or always negative in which case the optimal answer is infinity. However, in the case when $c_2^2 \ge 4c_1c_3$, the solution is $\max(0, \frac{-c_2+\sqrt{c_2^2-4c_1c_3}}{2c_1})$. Thus,
%\begin{equation}
%p_n^k = \begin{cases}
%I(c_3), \quad  c_2^2<4c_1c_3\\
%\max\left(0, \frac{-c_2+\sqrt{c_2^2-4c_1c_3}}{2c_1}\right), \quad c_2^2\ge 4c_1c_3
%\end{cases}\label{updatepnk}
%\end{equation}

\noindent {\bf Problem 2: Updating ${\cal L}$ }

Let $\beta_{n}^{k,v} = \sum_{t=1, t\ne k}^{v} l_n^t +\sum_{t=1}^{v}\sum_{m\in {\cal N}, m\ne n} r_{n,m}^t -\sum_{t=1}^{v}\sum_{m\in {\cal N}} r_{m,n}^t +\sum_{t=1}^{v}s_n^t+\sum_{t=1}^{v}D_n^t-E_n^v$ for $k\le v\le K$ and $n\in {\cal N}$.

\begin{eqnarray}
&&l_n^k=\argmin_{l_n^k \geq 0} \left(\left(\sum_{v=k}^{K}\left(y_{1,n}^v +y_{2,n}^v\right) - y_{3,n}^k\right)l_n^k \right.\nonumber\\&&+\frac{\rho}{2} \sum_{v=k}^{K}\left(l_n^k+ \beta_{n}^{k,v}+ u_{1,n}^v\right)^2
   \nonumber \\ && +\frac{\rho}{2} \sum_{v=k}^{K}\left(l_n^k+ \beta_{n}^{k,v} - u_{2,n}^v+B_n^{\max}\right)^2\nonumber\\&& \left.+ \frac{\rho}{2}\left(p_n^k - l_n^k - s_n^k g_n^k\right)^2+\frac{1}{2}\tau(l_n^k-(l_n^k)^{i})^2\right).
\end{eqnarray}
This is a quadratic equation. By differentiating it we obtain
%\begin{figure*}
\begin{eqnarray}
l_n^k=\frac{1}{\rho(2(K-k)+3)+\tau}\max \left(0,\tau (l_n^k)^{i} \quad \quad \quad \quad \right.\nonumber\\
- \left(\sum_{v=k}^{K}\left(y_{1,n}^v +y_{2,n}^v\right) - y_{3,n}^k +\rho \sum_{v=k}^{K}  \left( \beta_{n}^{k,v}+ u_{1,n}^v\right) \quad\right.\nonumber \\ 
 \left.\left.  +\rho \sum_{v=k}^{K}  \left(\beta_{n}^{k,v} - u_{2,n}^v+B_n^{\max} \right) - \rho  \left( p_n^k  - s_n^k- g_n^k\right)\right)\right).\label{lnkupdate}
\end{eqnarray}
%\end{figure*}
%Expanding the square terms and canceling all terms that are not function of $l_n^k$ yields,\\
%\begin{multline}
%l_n^k = \argmin_{l_n^k \geq 0} ((y_{1,n}^k) +(y_{2,n}^k) - y_{3,n}^k)l_n^k+\frac{\rho}{2} ( l_n^k^2+2l_n^k\beta_n^k)\\
%  + \frac{\rho}{2} (   l_n^k^2+2l_n^k\beta_n^k+2l_n^k B_n^{\max})
%  +  \frac{\rho}{2} (-2l_n^kp_n^k + l_n^k^2 +2l_n^ks_n^k +2l_n^k g_n^k)
%\end{multline}

%\begin{equation}
%l_n^k = \frac{1}{2\rho(K-k+1)+\rho}\max(0, -\sum_{v \geq k}y_{1,n}^v -\sum_{v \geq k}y_{2,n}^v+y_{3,n}^k+\rho( p_n^k-g_n^k-\sum_{v \geq k}X_{n,k}^v))
%\end{equation}

\noindent {\bf Problem 3: Updating ${\cal R}$ }

For $m=n$, $r_{m,n}^k = 0$. Otherwise, let $\nu_{m,n}^{k,v} = \sum_{t=1}^{v} l_m^t +\sum_{t=1, \cdots,v, b\in {\cal N}, (t,b)\ne (k,n), b\ne m} r_{m,b}^t -\sum_{t=1}^{v}\sum_{b\in {\cal N}} r_{b,m}^t +\sum_{t=1}^{v}s_m^t+\sum_{t=1}^{v}D_m^t-E_m^v$, and $\gamma_{m,n}^{k,v} = \sum_{t=1}^{v} l_n^t +\sum_{t=1}^{v}\sum_{n\in {\cal N}} r_{n,b}^t -\sum_{t=1, \cdots,v, b\in {\cal N}, (t,b)\ne (k,m), b\ne n} r_{b,n}^t +\sum_{t=1}^{v}s_n^t+\sum_{t=1}^{v}D_n^t-E_n^v$ for $k\le v\le K$ and $n,m\in {\cal N}, n \ne m$. Then, the optimization for ${\cal R}$ reduces as

\begin{eqnarray}
&&r_{m,n}^k=\argmin_{r_{m,n}^k \geq 0}\left(\left(\mu+\sum_{v=k}^{K}\left(y_{1,m}^v - y_{1,n}^v+y_{2,m}^v - y_{2,n}^v\right) \right.\right.\nonumber\\
&&\left.-y_{5,n}^k\right)r_{m,n}^k+\frac{\rho}{2} \sum_{v=k}^{K}\left( r_{m,n}^k +\nu_{m,n}^{k,v}+ u_{1,m}^v\right)^2 \nonumber \\ 
&& +\frac{\rho}{2} \sum_{v=k}^{K} \left( -r_{m,n}^k +\gamma_{m,n}^{k,v}+ (u_{1,n}^v)\right)^2\nonumber\\&&
+\frac{\rho}{2} \sum_{v=k}^{K}\left(r_{m,n}^k +\nu_{m,n}^{k,v}- u_{2,m}^v+B_m^{\max}\right)^2  \nonumber \\ && +\frac{\rho}{2} \sum_{v=k}^{K}\left(-r_{m,n}^k +\gamma_{m,n}^{k,v}- u_{2,n}^v+B_n^{\max}\right)^2
\nonumber\\&& +\frac{\rho}{2}\left(s_n^k +u_{4,n} - \sum_{b\in {\cal N}, b\ne n, b\ne m} r_{b,n}^k - r_{m,n}^k\right)^2\nonumber\\&&\left.+\frac{1}{2}\tau(r_{m,n}^k-(r_{m,n}^k)^{i})^2\right).
\end{eqnarray}

This is a quadratic equation. By differentiating it we obtain

\begin{eqnarray}
&&r_{m,n}^k =\frac{1}{\rho(4(K-k)+5)+\tau}\max \left(0,\tau(r_{m,n}^k)^{i}-\left(\mu+\right. \right. \nonumber\\
&&\sum_{v=k}^{K}\left(y_{1,m}^v - y_{1,n}^v+y_{2,m}^v - y_{2,n}^v\right) -y_{5,n}^k \nonumber\\
&&+\rho \sum_{v=k}^{K} \left(\nu_{m,n}^{k,v}+ u_{1,m}^v-\gamma_{m,n}^{k,v}- u_{1,n}^v\right)\nonumber\\&&+\rho \sum_{v=k}^{K} \left(\nu_{m,n}^{k,v}- u_{2,m}^v+B_m^{\max}-\gamma_{m,n}^{k,v}+ u_{2,n}^v-B_n^{\max}\right) \nonumber \\ && \left. \left.
 -\rho   \left( s_n^k +u_{4,n} - \sum_{b\in {\cal N}, b\ne n, b\ne m} r_{b,n}^k \right)\right) \right).\label{updatermn}
\end{eqnarray}

\noindent {\bf Problem 4: Updating ${\cal S}$ }

Let $\beta_{n}^{k,v} = \sum_{t=1}^{v} l_n^t +\sum_{t=1}^{v}\sum_{m\in {\cal N}, m\ne n} r_{n,m}^t -\sum_{t=1}^{v}\sum_{m\in {\cal N}} r_{m,n}^t +\sum_{t=1, t\ne k}^{v}s_n^t+\sum_{t=1}^{v}D_n^t-E_n^v$ for $k\le v\le K$ and $n\in {\cal N}$.

\begin{eqnarray}
&&s_n^k=\argmin_{s_n^k \geq 0} \left(\left(\sum_{v=k}^{K}\left(y_{1,n}^v +y_{2,n}^v\right) - y_{3,n}^k+ y_{5,n}^k\right)s_n^k \right.\nonumber\\
&&+\frac{\rho}{2} \sum_{v=k}^{K}\left(s_n^k+ \beta_{n}^{k,v}+ u_{1,n}^v\right)^2
  \nonumber \\&&+\frac{\rho}{2} \sum_{v=k}^{K}\left(s_n^k+ \beta_{n}^{k,v} - u_{2,n}^v+B_n^{\max}\right)^2\nonumber\\ 
  &&  +\frac{\rho}{2}(p_n^k - l_n^k - s_n^k- g_n^k)^2 +\frac{1}{2}\tau(s_n^k-(s_n^k)^{i})^2 \nonumber\\&&\left.+\frac{\rho}{2}\left(s_n^k +u_{4,n} - \sum_{m\in {\cal N}, m\ne n} r_{m,n}^k\right)^2 \right).
\end{eqnarray}
This is a quadratic equation. By differentiating it we obtain

\begin{eqnarray}
&&s_n^k=\frac{1}{\rho(2(K-k)+4)+\tau}\max \left(0,\tau (s_n^k)^{i}\right.  \nonumber\\
&&- \left(\sum_{v=k}^{K}\left(y_{1,n}^v +y_{2,n}^v\right)  +\rho \sum_{v=k}^{K}  \left( 2\beta_{n}^{k,v}+ u_{1,n}^v - u_{2,n}^v  \right. \right.\nonumber\\
&&\left.+B_n^{\max}\right)- y_{3,n}^k+ y_{5,n}^k -  \rho  \left(p_n^k - l_n^k - g_n^k\right)
  \nonumber \\&&  \left. \left.+\rho   \left(u_{4,n} - \sum_{m\in {\cal N}, m\ne n} r_{m,n}^k\right)\right)\right).\label{snkupdate}
\end{eqnarray}

\noindent {\bf Problem 5: Updating ${\cal G}$ }

\begin{eqnarray}
g_n^k&=&\argmin_{g_n^k \geq 0} \left(\lambda g_n^k -y_{3,n}^kg_n^k+ \frac{\rho}{2} \left( p_n^k - l_n^k - s_n^k -g_n^k\right)^2\right. \nonumber\\
&&\left. +\frac{1}{2}\tau(g_n^k-(g_n^k)^{i})^2\right).
\end{eqnarray}
This is a quadratic equation. By differentiating it we obtain

\begin{equation}
g_n^k = \frac{1}{\rho+\tau}\max \left(0, -\lambda +y_{3,n}^k+ \rho\left(p_n^k- l_n^k-s_n^k\right)+\tau (g_n^k)^{i} \right).\label{gnkupdate}
\end{equation}

\noindent {\bf Problem 6: Updating ${\cal D}$ }

Let $\beta_{n}^{k,v} = \sum_{t=1}^{v} l_n^t +\sum_{t=1}^{v}\sum_{m\in {\cal N}, m\ne n} r_{n,m}^t -\sum_{t=1}^{v}\sum_{m\in {\cal N}} r_{m,n}^t +\sum_{t=1}^{v}s_n^t+\sum_{t=1,t\ne k}^{v}D_n^t-E_n^v$ for $k\le v\le K$ and $n\in {\cal N}$.

\begin{eqnarray}
&&D_n^k=\argmin_{D_n^k \geq 0} \left(\left(\sum_{v=k}^{K}\left(y_{1,n}^v +y_{2,n}^v\right)\right)D_n^k \right. \nonumber\\
&&+\frac{\rho}{2} \sum_{v=k}^{K}\left(D_n^k+ \beta_{n}^{k,v}+ u_{1,n}^v\right)^2+\frac{1}{2}\tau(D_n^k-(D_n^k)^{i})^2\nonumber\\
&&\left. +\frac{\rho}{2} \sum_{v=k}^{K}\left(D_n^k+ \beta_{n}^{k,v} - u_{2,n}^v+B_n^{\max}\right)^2\right).
\end{eqnarray}

This is a quadratic equation. By differentiating it we obtain

\begin{eqnarray}
D_n^k=\frac{1}{2\rho(K-k+1)+\tau}\max \left(0,-\left( \sum_{v=k}^{K}\left(y_{1,n}^v +y_{2,n}^v\right)\right.\right. \nonumber\\\left.\left.+\rho \sum_{v=k}^{K}  \left( 2\beta_{n}^{k,v}+ u_{1,n}^v - u_{2,n}^v+B_n^{\max}\right)\right)+\tau (D_n^k)^{i}\right).\label{updatednk}
\end{eqnarray}

\noindent {\bf Problem 7: Updating ${\cal U}$ }

The optimization for each of $u_{i,n}^k$ is a quadratic problem, and thus the solutions for these problems are as follows.

\begin{eqnarray}
u_{1,n}^k &=& \frac{1}{\rho+\tau}\max\left(0,-y_{1,n}^k -\rho \sum_{t=1}^{k} l_n^t -\rho \sum_{t=1}^{k} \right.\nonumber\\&&\sum_{m\in {\cal N}, m\ne n} r_{n,m}^t +\rho\sum_{t=1}^{k}\sum_{m\in {\cal N}} r_{m,n}^t -\rho\sum_{t=1}^{k}s_n^t\nonumber\\&&\left. -\rho\sum_{t=1}^{k}D_n^t+\rho E_n^k+\tau (u_{1,n}^k)^{i}\right)\label{updateunk1},\\
u_{2,n}^k &=& \frac{1}{\rho+\tau}\max\left(0,y_{2,n}^k +\rho \sum_{t=1}^{k} l_n^t +\rho \sum_{t=1}^{k}\right.\nonumber\\&&
\sum_{m\in {\cal N}, m\ne n} r_{n,m}^t -\rho \sum_{t=1}^{k}\sum_{m\in {\cal N}} r_{m,n}^t +\rho \sum_{t=1}^{k}s_n^t\nonumber\\&&\left.+\rho \sum_{t=1}^{k}D_n^t-\rho E_n^k+\rho B_n^{\max}+\tau (u_{2,n}^k)^{i}\right),\\
u_{3,n}^k &=& \frac{1}{\rho+\tau}\max\left(0,-y_{4,n}^k -\rho p_n^k+\rho P_n\right.\nonumber\\&&\left.+\tau (u_{3,n}^k)^{i}\right),\\
u_{4,n}^k &=& \frac{1}{\rho+\tau}\max\left(0,-y_{5,n}^k -\rho s_n^k \right.\nonumber\\&&\left.+ \rho \sum_{m\in {\cal N}, m\ne n} r_{m,n}^k+\tau (u_{4,n}^k)^{i}\right).\label{updateunk2}
\end{eqnarray}

\end{document}